\documentclass[11pt,a4paper]{article}

\usepackage[margin=1in]{geometry}
\usepackage[T1]{fontenc}

\usepackage{mathtools}
\usepackage{amsthm}
\usepackage{amssymb}

\usepackage[tt=false]{libertine}
\usepackage[libertine]{newtxmath}
\usepackage{bm}
\usepackage{bbm}
\usepackage{microtype}

\usepackage[ruled]{algorithm2e}

\usepackage{enumitem}

\usepackage[authoryear,square]{natbib}
\usepackage[svgnames]{xcolor}
\usepackage{hyperref}
\hypersetup{
  colorlinks=true,
  urlcolor=blue,
  linkcolor=DarkBlue,
  citecolor=DarkGreen
}

\usepackage{doi}

\usepackage{tikz}
\usetikzlibrary{positioning,calc,arrows.meta}

\usepackage[capitalise,nameinlink,noabbrev]{cleveref}

\usepackage{xspace}

\renewcommand{\vec}{\bm}

\newlist{conditions}{enumerate}{1}
\setlist[conditions]{label={(\arabic*)}, ref={\arabic*}}
\crefname{conditionsi}{Condition}{Conditions}
\Crefname{conditionsi}{Condition}{Conditions}

\crefname{algocf}{Algorithm}{Algorithms}
\Crefname{algocf}{Algorithm}{Algorithms}

\usepackage{caption}
\usepackage{pgfplots}
\pgfplotsset{compat=1.18}

\usetikzlibrary{decorations.pathreplacing}

\def \R{\mathbb R}

\newcommand{\sset}[1]{\left\{ #1\right\}}
\newcommand{\ssets}[1]{\{ #1\}}

\newcommand{\fwhs}[1]{\; | \; #1 }
\newcommand{\card}[1]{\left| #1 \right|}
\newcommand{\cards}[1]{| #1 |}
\newcommand{\union}{\cup}
     
\newcommand{\map}{\to}
 
\newcommand{\inters}{\cap}    
    
\newcommand{\then}{\Longrightarrow} 

\newcommand{\bigland}{\bigwedge}

\DeclareMathOperator*{\expectation}{\mathbb E}
\newcommand{\expect}[2][]{\expectation_{#1}\nolimits\left[#2\right]}

\DeclareMathOperator*{\probability}{\mathrm{Pr}}
\newcommand{\prob}[1]{\probability\left[#1\right]}
\newcommand{\probs}[1]{\probability[#1]}

\DeclareMathOperator*{\argmax}{argmax}

\DeclareMathOperator*{\supportdistro}{\mathrm{supp}}
\newcommand{\support}[1]{\supportdistro\left(#1\right)}

\newcommand*{\poly}{\operatorname{poly}}

\newcommand{\eps}{\varepsilon}
\renewcommand{\epsilon}{\eps}

\newcommand{\compstrat}{\textsf{\textup{ComputeStrategy}}\xspace}

\newcommand{\deltafunc}{\Delta}

\newcommand{\powidx}{\mu} %% Internal variable for the CCFPA explicit upper bound 

\theoremstyle{definition}
\newtheorem{definition}{Definition}

\theoremstyle{plain}
\newtheorem{theorem}{Theorem}[section]
\newtheorem{lemma}[theorem]{Lemma}
\newtheorem{corollary}[theorem]{Corollary}
\newtheorem{proposition}[theorem]{Proposition}

\Crefname{claim}{Claim}{Claims}
\newtheorem{inftheorem}{Informal Theorem}
\newtheorem{property}{Property}

\theoremstyle{remark}
\newtheorem{remark}{Remark}

\newcommand*{\myproofname}{Proof}
\newenvironment{nestedproof}[1][\myproofname]{\begin{proof}[#1]}{\end{proof}}

\title{Efficient Equilibrium Computation in\\ Symmetric First-Price Auctions\thanks{Aris Filos-Ratsikas was supported by the UK Engineering and Physical Sciences Research Council (EPSRC) grant EP/Y003624/1. Charalampos Kokkalis was supported by an EPSRC DTA Scholarship (Reference EP/W524384/1).}}

\author{
\begin{tabular}{c c}
& \\ \textbf{Aris Filos-Ratsikas} & \textbf{Yiannis Giannakopoulos}\\
\small{University of Edinburgh, United Kingdom} & \small{University of Glasgow, United Kingdom} \\
\href{mailto:Aris.Filos-Ratsikas@ed.ac.uk}{\small{\texttt{aris.filos-ratsikas@ed.ac.uk}}} & \href{mailto:yiannis.giannakopoulos@glasgow.ac.uk}{\small{\texttt{yiannis.giannakopoulos@glasgow.ac.uk}}}\\
& \\
\textbf{Alexandros Hollender} & \textbf{Charalampos Kokkalis}\\
\small{University of Oxford, United Kingdom} & \small{University of Edinburgh, United Kingdom} \\
\href{mailto:alexandros.hollender@cs.ox.ac.uk}{\small{\texttt{alexandros.hollender@cs.ox.ac.uk}}} & \href{mailto:charalampos.kokkalis@ed.ac.uk}{\small{\texttt{charalampos.kokkalis@ed.ac.uk}}}
\end{tabular}}

\date{July 27, 2026}

\begin{document}
\maketitle

\begin{abstract}
We study the complexity of computing Bayes-Nash equilibria in single-item first-price auctions. We present the first \emph{efficient} algorithms for the problem, when the bidders' values for the item are independently drawn from the same continuous distribution. More precisely, we design polynomial-time algorithms for the white-box model, where the distribution is provided directly as part of the input, and query-efficient algorithms for the black-box model, where the distribution is accessed via oracle calls. Our results settle the computational complexity of the problem for bidders with i.i.d.\ values.
\end{abstract}

\newpage

\section{Introduction}
\label{sec:intro}

First-price auctions have been used since antiquity for the sale of various commodities such as land, property, farming rights, livestock, or tax and public procurement contracts \citep{cassady1967auctions}. Today, these auctions are widely used in markets for advertising space, in the form of ``ad auctions'' (see, e.g., \citep{balseiro2021robust,despotakis2021first,paes2020competitive,conitzer2022multiplicative,alcobendas2022adjustment}), constituting a multi-billion-dollar industry \citep{bowman2024much}. Part of the appeal of these auctions is their simplicity and intuitive nature: the winner is the bidder with the highest bid (breaking ties uniformly at random), who pays an amount equal to her bid.

The development of the theory of first-price auctions dates back to the early 1960s and the pioneering work of \citet{vickrey1961counterspeculation}. Vickrey studied the \emph{equilibria} of the auction---or, more precisely, of the strategic game induced by the incentives of the bidders to underbid their values for the item---when values are independently drawn from the uniform distribution. Over the next decades, several variants of the auction were studied systematically in a plethora of works in economics, e.g., see \citep{griesmer1967toward,riley1981optimal,MW82,chwe1989discrete,plum1992characterization,marshall1994numerical,lebrun1996existence,lebrun1999first,maskin2000equilibrium,lizzeri2000uniqueness,Athey2001,reny2004existence}. Most of these works were mainly concerned with establishing the existence, and often the uniqueness, of equilibria of the auction. For the simpler case of \emph{symmetric} first-price auctions (most notably for i.i.d.\ values), the literature has managed to produce elegant analytic solutions, e.g., see \citep{vickrey1961counterspeculation,riley1981optimal,MW82}. Symmetric auctions capture cases where the bidders' information is generated by the same technologies or market parameters, and in fact constitute the ``textbook'' model studied in the auction literature \citep{myerson1981optimal,mcafee1987auctions,krishna2009auction}. 

While the economic theory established the existence of equilibria and managed, via the aforementioned analytic formulas, to \emph{describe} them for symmetric variants, it did not address the question of how these equilibria can be formally \emph{computed}. The computation of equilibria in games has been a focal point of research in computer science over the past two decades, e.g., see \citep{daskalakis2009complexity,chen2009settling,mehta2018constant,rubinstein2018inapproximability}. For the first-price auction, this question was formally put forward by \citet{fghlp2021_sicomp}, who studied the auction with a \emph{continuous value space} and a \emph{discrete bidding space}. This model, coined the \emph{Continuous First-Price Auction (CFPA)}\footnote{This name was given by \citet{fghk24}, to distinguish this model from the First-Price auction with discrete values and discrete bids, which they also study. Earlier works, namely \citep{fghlp2021_sicomp,chen2023complexity} referred to this auction as simply the \emph{First-Price Auction (FPA)}. While we do not consider variants with discrete value spaces in our work, we adopt the terminology of \citet{fghk24} for consistency with the more recent literature.} is particularly appealing, because it maintains the attractive equilibrium existence properties (e.g., see \citep{Athey2001}), but it is also very appropriate for computational purposes. 

A series of works, including and following \citep{fghlp2021_sicomp}, provided computational hardness results for computing equilibria of the auction. However, these results only concern variants of the auction with rather general assumptions on the value distributions, namely \emph{subjectivity} \citep{fghlp2021_sicomp,fghk24} or \emph{arbitrary correlation} \citep{fghk2025}, or variants with \emph{non-standard tie-breaking rules} \citep{chen2023complexity}. On the other hand, the positive results obtained in these works, even for simple symmetric settings like auctions with i.i.d.\ values, are somewhat limited: currently, the best known result for these settings is a polynomial-time approximation scheme (PTAS), i.e., an algorithm which computes an $\varepsilon$-approximate equilibrium in time polynomial in the input parameters, but possibly exponential in $1/\varepsilon$, see \citep{fghlp2021_sicomp,fghk2025}. This leaves us with the following important question for the fundamental setting of symmetric auctions: \emph{``Can we design efficient algorithms for computing their equilibria?''.}

In the quest to answer this question, we immediately encounter a first major challenge: For settings with discrete bids (i.e., the CFPA), the literature has not in general been able to provide analytic formulas that describe the equilibria of the auction; indeed, such formulas are only available for continuous bidding spaces. Instead, the classic results in the literature establish merely the \emph{existence} of equilibria of the CFPA via applications of fixed-point theorems, like Brouwer's or Kakutani's fixed-point theorem (e.g., see \citep{Athey2001,fghlp2021_sicomp}). As a result, these existence proofs do not provide any real insights for the design of efficient algorithms for equilibrium computation. Therefore, in order to design such algorithms, we essentially need to come up with new, \emph{constructive} proofs for establishing equilibrium existence.  \medskip 

\noindent Our main result in this paper is the \emph{first efficient algorithm} for computing equilibria in the symmetric CFPA.

\begin{inftheorem}\label{infthm:main-cfpa}
A (symmetric) equilibrium of the symmetric CFPA can be computed efficiently.
\end{inftheorem}

\noindent We consider the result of \cref{infthm:main-cfpa} to be quite important in the context of this literature. On the one hand, it improves significantly over the state-of-the-art algorithms of \citet{fghlp2021_sicomp,fghk2025} for symmetric settings. Additionally, it categorically establishes that any attempts at showing computational hardness results for the equilibrium computation problem would have to be on non-symmetric variants of the auction, i.e., the general \emph{independent private values (IPV)} setting, e.g., see \citep{MW82}. 

\paragraph{Auctions with Continuous Bids.} While in the context of equilibrium computation, the CFPA is the most established variant of the auction, it is still instructive to consider the variant where both the value space and the bidding space are \emph{continuous}. This setting was coined the \emph{Continuous First-Price Auction with Continuous Bids (CCFPA)} by \citet{fghk2025}, who studied it merely as a technical tool for proving their results about the CFPA. 

Continuous bidding spaces are most often used as a convenient mathematical abstraction of real-world bidding, which naturally comes in discrete increments \citep{cox1988theory,rothkopf1994role,david2007optimal}. Indeed, this continuity has allowed researchers to come up with the analytic formulas that we mentioned earlier to describe the equilibria. For the case of symmetric settings, these formulas are due to \citet{riley1981optimal} and \citet{MW82}. From a computational complexity perspective, it is very natural to ask whether we can somehow leverage these analytic formulas to design efficient algorithms for computing equilibria. In principle, this seems like a less challenging task compared to the CFPA, where such formulas were not at our disposal.

Still, this endeavor is far from trivial, primarily because these analytic solutions are not readily appropriate for computational purposes. Hence, making use of them to obtain polynomial-time algorithms requires careful consideration of the representation of the inputs and outputs to the problem and a meticulous ``translation'' of the mathematical insights to computational methods. This is indeed conceptually simpler, but requires significant technical work. To this end, we have the following informal theorem, which closely mirrors the result of \cref{infthm:main-cfpa}.

\begin{inftheorem}\label{infthm:main-ccfpa}
A (symmetric) equilibrium of the symmetric CCFPA can be computed efficiently.
\end{inftheorem}

\noindent Prior to our work, results about the complexity of equilibrium computation in the CCFPA were not known in the literature. In addition, we believe that our techniques, which we highlight in \cref{sec:ccfpa-results} below, could be useful for establishing similar equilibrium computation results for other auction formats, besides the first-price auction. 

\subsection{Our Setting, Results, and Techniques}
\label{sec:our-results}
We study first-price auctions in which a set of $n$ bidders submit their bids for one item for sale. Each bidder knows her own value for the item, and has incomplete information about the values of her competitors, via a (common) continuous distribution, from which the values are independently drawn (i.i.d.). Our main focus is on the established variant of the Continuous First-Price Auction (CFPA), but we also consider the Continuous First-Price Auction with Continuous Bids (CCFPA). Our goal is to design \emph{efficient algorithms} that compute equilibria, and in fact, even \emph{symmetric equilibria} of the auction. These are (Bayes-Nash) equilibria in which every bidder employs the same strategy, i.e., the same mapping from values to bids. In fact, considering \emph{approximate} equilibria might be necessary, as exact equilibria might involve irrational numbers in their description \citep{fghlp2021_sicomp}, and hence cannot be found in standard models of computation.

To be more precise, when the distribution is given to the algorithm directly as part of the input, an algorithm is deemed ``efficient'' if it runs in polynomial time in the size of the input representation. This representation can be \emph{explicit}, i.e., via some succinctly described function (e.g., piecewise polynomial), or \emph{implicit}, where the input is some meaningful computational device (e.g., an appropriate arithmetic circuit) that the algorithm can use to compute the distribution's cdf. We also (for the first time in the literature of the problem) consider the \emph{black-box} model, where the algorithm has \emph{query} or \emph{oracle access} to the distribution. In this case, the efficiency of the algorithm is measured by the number of queries that it performs to the cdf oracle, and the number of arithmetic operations that it performs.

Furthermore, we will consider an algorithm to be efficient if it computes either an \emph{exact} equilibrium whenever possible (i.e., whenever the equilibrium can be described using rational numbers), or otherwise an \emph{approximate} equilibrium with the best possible accuracy (ideally, exponentially small error). An algorithm that computes an $\varepsilon$-approximate equilibrium in time polynomial in the input size and $1/\varepsilon$ will be referred to as a \emph{Fully Polynomial-Time Approximation Scheme (FPTAS)}, whereas a PTAS, mentioned earlier, runs in polynomial time for every \emph{fixed} $\varepsilon$ (but may be exponential in $1/\varepsilon$). For a more in-depth discussion of these (standard) approximation notions, the interested reader is referred to the classic textbook of~\citet{Vazirani2001a}. In the black-box model, we will deem an algorithm efficient if it computes an $\varepsilon$-approximate equilibrium using the (provably, qualitatively) smallest possible number of queries as a function of $\varepsilon$, as well as a polynomial number of arithmetic operations.  

\subsubsection{Algorithms for the CFPA.}\label{sec:cfpa-results}

As we mentioned above, the literature in economics has not in general managed to provide analytic solutions for the equilibria of the CFPA.
In order to come up with a general-purpose efficient algorithm for computing equilibria of the CFPA, we would need to provide a \emph{new, constructive proof of existence} of (approximate) equilibria, via an efficient computational method. 

This is precisely what we do in \cref{sec:cfpa}. By unfolding the details of \cref{infthm:main-cfpa} for this setting, we obtain an algorithm which computes an $\varepsilon$-approximate equilibrium, and 
\begin{itemize}
    \item[-] in the black-box model, it performs a number of queries which is polynomial in $\log(1/\varepsilon)$ and $\log(L)$, where $L$ is the Lipschitz parameter of the distribution, as well as a polynomial number of arithmetic operations, and 
    \item[-] in the model where the distribution is given in the input as an arithmetic circuit, it runs in time polynomial in $\log(1/\varepsilon)$ and $\log(L)$ and the other parameters of the input.
\end{itemize}
For the case of the explicit representation of the distribution via a piecewise polynomial function, the dependence on the Lipschitz parameter above vanishes, as it is ``absorbed'' in the description of the polynomial pieces. 

\paragraph{Overview of our techniques for the CFPA.} As we explain in detail in \cref{sec:prelims-bids}, we will be interested in computing \emph{monotone} equilibria, i.e., equilibria which assign higher bids to higher values, similarly to \citep{fghlp2021_sicomp,chen2023complexity,fghk24,fghk2025}. Such equilibria are known to exist in general first-price auctions \citep{Athey2001}, and computing them makes our positive results even stronger. It is known that monotone equilibria can be represented in a convenient and succinct way, by means of \emph{jump points} $0=s_0 \leq s_1 \leq \ldots \leq s_m=1$; informally, jump point $s_j$ denotes the value for which a bidder switches from using bid $b_{j}$ to $b_{j+1}$. Our goal will thus be to efficiently compute an equilibrium set of jump points. 

At the heart of our algorithm (\cref{alg:cdfpa-outer}) lies the following crucial observation: if we could find an equilibrium utility $U$ of any (due to symmetry) bidder when she has the highest possible value (i.e., $v=1$) for the item, then we would be able to efficiently compute a symmetric equilibrium. This is achieved by a subroutine which we refer to as \compstrat\ (see \cref{alg:cdfpa-inner}). At a high level, \compstrat computes the set of jump points ``from the top'', i.e., starting from $s_m$, and then computing any jump point $s_{i-1}$ after having computed the ``previous'' jump points $s_{m}, \ldots, s_{i}$. To do that, the algorithm maintains a \emph{candidate utility level} $U$ and three \emph{branches}, namely (a) the first (or ``if'') branch, (b) the second (or ``else if'') branch, and (c) the third (or ``else'') branch, see \cref{alg:cdfpa-inner} for reference. The ``normal'' computation branch is the third one: indeed, given the jump points $s_m,\ldots, s_i$, the algorithm performs binary search on the value space to find a position of jump point $s_{i-1}$ which guarantees that the expected utility of the bidder is very close to $U$. The proximity is determined by a parameter $\delta$ which depends on the equilibrium approximation parameter $\varepsilon$, and which determines the termination condition of the binary search.

Sometimes however, this ``normal'' computation of $s_{i-1}$ might fail. One potential reason is that at the equilibrium, $s_{i-1}$ might need to be ``skipped''; the appropriate condition for that is identified in the first branch, where $s_{i-1}$ is simply set to $s_i$, i.e., the previously computed jump point. Another, more significant failure mode is encountered in the second branch: in this case, the algorithm identifies that given the choice of $s_m,\ldots,s_{i}$, there is no way to choose the remaining jump points so that the expected utility of the bidder is very close to the candidate utility $U$. Essentially, this means that $U$ is not an equilibrium utility.  

The important point here is that in that case, i.e., when we run \cref{alg:cdfpa-inner} to obtain an equilibrium from a candidate utility of $U$ at $v=1$ which does not correspond to an equilibrium utility, the method will still return a set of jump points $(s_0,\ldots, s_m)$, just not equilibrium jump points. Still, these points convey important information: we show that (a) for $U=0$, we have that $s_0=0$, and (b) for $U=1$, we have that $s_0 > 0$. This suggests the following natural approach, which we adopt in \cref{alg:cdfpa-outer}: we perform binary search on $U$ in order to find a value $U^*$ such that $s_0$ is ``barely above'' $0$; indeed, we show that, given $U^*$, we can efficiently compute an approximate equilibrium. The most challenging part of the analysis is to show that our method is continuous, i.e., that the jump points that it outputs depend continuously on the value of $U$. We provide the algorithm and the full details of the analysis in \cref{sec:algorithms-description} and \cref{sec:algorithm-correctness}.  

\subsubsection{Algorithms for the CCFPA.}\label{sec:ccfpa-results} 
As opposed to the CFPA (see~\cref{sec:cfpa-results}), once we allow for continuous bidding spaces, seminal work from the economics literature provides elegant analytic formulas describing (exact) equilibria of CCFPA settings (see~\cref{sec:ccfpa}). This readily gives us a strong starting point in our quest to compute (approximate) equilibria in CCFPA. Nevertheless, turning this understanding to concrete approximation algorithms and characterizing the computational complexity of the problem turns out to be a non-trivial task, which has not been pursued before. More precisely, we are able to obtain the following results:
\begin{itemize}[leftmargin=*]
    \item[-] For the black-box model, in \cref{sec:black-box} we present an efficient algorithm which computes an $\varepsilon$-approximate equilibrium using a number of queries which is polynomial in $1/\varepsilon$, as well as a polynomial number of arithmetic operations (see~\cref{th:black-box-upper}). Importantly, we further complement this result with an asymptotically matching query-complexity lower bound, establishing that our algorithm is the best possible (\cref{th:black-box-lower}).\smallskip
    
    \item[-] Via a careful inspection of our analysis for the black-box model above, it readily follows that the same algorithm can be transformed into an FPTAS for the case when the distribution is provided in the input implicitly via an arithmetic circuit (see \cref{sec:white-box}). \smallskip

    \item[-] Finally, in \cref{sec:CCFPA-explicit} we consider the setting where the (cdf of the) distribution is explicitly provided in the input of our algorithms, via a piecewise polynomial function --- a standard representation model in the literature of the problem (see, e.g., \citep{fghlp2021_sicomp,fghk24,fghk2025}). In this case, we provide a different algorithm, which computes an \emph{exact} equilibrium in polynomial time, hence strengthening the FPTAS above, for this very natural special case. As an interesting by-product of our analysis for this case, we can provably establish that there exists an exact equilibrium with rational bids on rational input values. 
\end{itemize}

\paragraph{Overview of our techniques for the CCFPA.}
As mentioned above, for CCFPA we already have access to simple analytic formulas that describe the ``canonical'', monotone exact equilibrium of the auction (see~\eqref{eq:CCFPA-IID-BNE-beta}). Therefore, it is tempting to think that any naive numerical-analysis approach of simply discretizing and approximating this continuous equilibrium-bidding function is all that is needed to provide the desired algorithm for our case. However, it turns out that there are a few intricate obstacles that one needs to handle under this approach, including: the requirement that, for \emph{any} desired approximation guarantee $\varepsilon>0$ for our equilibrium, our underlying computed object needs to satisfy various requirements, including no-overbidding and continuity, but most notably, monotonicity. Additionally, we need to make sure that our computed object is not just ``near enough'' to the ``ideal'' analytic function goal, but rather that it is provably an $\varepsilon$-approximate equilibrium as well. Our proof of~\cref{th:black-box-upper} takes great care into maintaining these properties; for example, observe that our upper/lower Riemann-sum approximation in \cref{eq:CCFPA-black-box-lower-riemann,eq:CCFPA-black-box-upper-riemann,eq:CCFPA-black-box-x-specific-discretization,eq:CCFPA-black-box-beta-riemann-bounds} is not naively taken in an equidistant-discretization way, but instead the last point $\hat{a}_{k_x+1}(x)=x$ needs to be ``variable'' and, furthermore, the upper sum is taken as our output object (see~\eqref{eq:CCFPA-black-box-FPTAS-tilde-beta-def}) and not the lower one (otherwise, monotonicity and no-overbidding will not be ensured).

Our matching lower-bound construction in~\cref{th:black-box-lower} is a ``classic'' information-theoretic, query-complexity construction: we explicitly construct two value distributions (see~\cref{fig:black-box-lower-bound}) that: (1) are indistinguishable with fewer than $\varOmega(1/\varepsilon)$ cdf queries, but at the same time (2) no $\varepsilon$-approximate equilibrium with respect to the first one can be a valid $\varepsilon$-approximate equilibrium of the other.

Finally, for the case where we a priori set a piecewise-polynomial input representation language for our value distributions, things become more concrete analytically: in principle, the problem boils down to carefully deriving and plugging-in these closed-form distribution formulas into the exact equilibrium expressions from the economics literature~\eqref{eq:CCFPA-IID-BNE-beta}. However, this high-level idea does not readily translate into an (efficient) \emph{algorithmic} way for arriving at an equilibrium. This is exactly what we focus on in~\cref{sec:CCFPA-explicit} where we provide a systematic way to recursively compute, in polynomial time, a set of parameters (whose magnitude's bit-size is shown to be polynomially bounded) that can be then used to express the desired equilibrium bidding-function, succinctly, as a ratio of two polynomials (see~\eqref{eq:ccfpa-poly-time-explicit-rational-expression}).   

\subsection{Related Work and Discussion}
The systematic study of equilibrium computation in first-price auctions was initiated by \citet{fghlp2021_sicomp}, where the CFPA was formally introduced and studied. The main result of this work is the computational hardness (specifically PPAD-hardness and FIXP-hardness, see \citep{etessami2010complexity}) of computing equilibria of the auction for distributions that are not only different for every agent, but also \emph{subjective}, i.e., different from the \emph{perspective} of each agent. \citet{chen2023complexity} subsequently showed a similar PPAD-hardness result without the subjectivity assumption, but for an involved (trilateral) tie-breaking rule, rather than the standard uniform tie-breaking that is typically used in the literature of the auction, e.g., see \citep{vickrey1961counterspeculation,riley1981optimal,MW82,Athey2001,krishna2009auction}. \citet{fghk24,fghk2025} showed corresponding NP-hardness results for a variant of the auction with discrete distributions, under the same subjectivity assumption, or for distributions that are arbitrarily correlated between bidders. 

In terms of positive results, \citet{fghlp2021_sicomp} provided a polynomial-time algorithm for the case where both the number of bidders and the number of bids is fixed, which employs off-the-shelf results of \citet{grigor1988solving} for approximately solving systems of polynomial inequalities. Using a ``bidding shrinkage lemma'' due to \citet{chen2023complexity}, \citet{fghk24,fghk2025} observed that the same algorithm can be transformed into a PTAS for a fixed number of bidders, and any bidding space. Furthermore, the dependency on the number of bidders can be removed for symmetric auctions, yielding a PTAS that holds under no further assumptions. As we explained earlier, a PTAS is only efficient when the approximation parameter $\varepsilon$ is constant; in contrast, our algorithm in \cref{sec:cfpa} achieves the best possible (inversely exponential) approximation for the same setting.

As we mentioned earlier, constructive proofs of equilibrium existence for the CFPA were not generally known in the literature before our work. To the best of our knowledge, the only known result of this flavour is due to \citet{chwe1989discrete}, who studies the equilibria of the first-price auction under two crucial assumptions, namely that the (discrete) bids are \emph{equidistant}, and that the values are \emph{uniformly distributed}. Chwe's equilibrium existence result can be interpreted as a \emph{non-computational} method to describe an equilibrium of the auction in this special case. Our work goes beyond the result of \citet{chwe1989discrete}; it does not impose any restrictions on the distributions or the structure of the bidding space, and at the same time provides an efficient algorithm for computing approximate equilibria of the auction. 

To the best of our knowledge, the complexity of equilibrium computation in the CCFPA had not been studied prior to our work. This variant of the auction was studied only as a technical tool by \citet{fghk2025}, as part of their ``bid densification'' approach: in a quest to achieve positive results for the CFPA, the authors consider how its equilibria would be approximated by equilibria of the CCFPA, for which analytic solutions are known (by the seminal work of~\citet{MW82}). Despite this connection, the authors did not show how to compute equilibria of the CCFPA in polynomial time. Furthermore, their resulting algorithm for the CFPA, although displaying a graceful dependency on $\varepsilon$, crucially also depends on other, non-controllable parameters such as explicit bounds on the pdf of the distribution and the granularity of the bidding space. For our results in \cref{sec:cfpa} we do not impose any such assumptions (other than a mild Lipschitz-continuity assumption).

In a complementary direction, \citet{ahunbay2024uniqueness} study structural properties of equilibria in standard first-price auctions, showing uniqueness of Bayesian coarse correlated equilibria (and extending the analysis to all-pay auctions).
There is also a growing line of work that employs tools from online learning and optimization to characterize (and in some cases compute) stable outcome notions in first-price auctions, particularly in repeated and non-stationary environments \citep{Feng2021NoRegretAuctions,bichler2023dual,wen2025jointvalueestimationbidding,hu2025learningbidnonstationaryrepeated}.

\section{Preliminaries}
\label{sec:prelims}

\subsection{Symmetric Auctions}
\label{sec:auctions}
In a \emph{symmetric (continuous, Bayesian) first-price auction}, there is a set $N=[n]$ of bidders and one item for sale (where $[n]\coloneq\{1,2,\ldots,n\}$). There is a common set of values $V = [0,1]$ (the \emph{value space}) and a set of bids $B \subseteq [0,1]$ (the \emph{bidding space}); each bidder $i$ has a value $v_i \in V$ for the item and submits a bid $b_i \in B$.
Similarly to previous work \citep{fghk24,chen2023complexity,fghk2025}, we assume that $0 \in B$, which can be interpreted as the option to abstain from the auction. 

Given a \emph{bid profile} $\vec{b}=(b_1,\ldots,b_n)$, the item is sold to one of the highest bidders $i^{*} \in W(\vec{b})\coloneq \argmax_{i \in N}b_i$ for a payment equal to her bid $b_{i^*}$, breaking ties \emph{uniformly at random}. 
We define the \emph{ex-post} utility of bidder $i$ with (true) value $v_i$, with respect to submitted bids $\vec{b}$, by:

\begin{equation}
\label{eq:ex_post_utilities}
\tilde{u}_i(\vec{b};v_i) \coloneq 
\begin{cases}
\frac{1}{\cards{W(\vec{b})}}(v_i-b_i), & \text{if}\;\; i\in W(\vec{b}), \\
0, & \text{otherwise}, 
\end{cases}
\end{equation}

\paragraph{Value Distributions.} The value of each bidder is private information to her only; bidders have \emph{incomplete information} about the values of others. 
More precisely, we assume that, from the perspective of some bidder $i$, the values of the other bidders for the item are i.i.d.\ random variables, drawn from an (absolutely) continuous distribution (with cdf) $F$ over $V=[0,1]$.  
We let $f$ denote the distribution's pdf and $\support{F}$ its closed support. Note that, without loss of generality, $\support{F}$ can be assumed to be a union of (countably many, disjoint)
subintervals of $[0,1]$. 
We use $\underline{v}\coloneqq \inf \support{F}$ to denote the leftmost point of $F$'s support; the distribution $F$ that $\underline{v}$ refers to will be clear from the context. 

\paragraph{Bidding Strategies.} A \emph{(pure, bidding) strategy} of a bidder $i$ is a function $\beta: V \map B$ mapping values to bids. 
A \emph{strategy profile} is a collection of strategies, one for each bidder, i.e., a vector $\vec{\beta}=(\beta_1,\ldots,\beta_n) \in B^V \times \ldots \times B^V$. 
In this work, we will only be interested in \emph{symmetric} strategy profiles, in which $\beta_i = \beta_j$ for all $i,j \in N$; in this case, we will slightly abuse notation and use $\beta \in B^V$ to denote the symmetric strategy profile $\vec{\beta} = (\beta,\ldots,\beta)$. 

\paragraph{Utilities and Equilibria.} Given a symmetric strategy profile $\beta$, the (interim) \emph{utility} of any bidder $i$ with (true) value $v \in V$, when she bids $b \in B$, is given by
\begin{align}
    u(b,\beta;v) 
    &\coloneq\expect[\vec{v}_{-i}\sim \times_{j=1}^{n-1} F]{\tilde{u}_i(b,\beta(\vec{v}_{-i});v)}, \label{def:utility-interim}
\end{align}
where $\beta(\vec{v}_{-i})$ denotes the vector $\times_{j\in N\setminus{\ssets{i}}}\beta(v_j)$. 
Notice that the subscript $i$ is deliberately omitted (that is, we simply write $u$ instead of $u_i$), due to the symmetry of definition~\eqref{def:utility-interim}. Whenever the strategy $\beta$ is clear from the context, we will drop it from the notation of the utility and simply write $u(b;v)$ instead.

\medskip
In this paper, we consider the following standard notion of (approximate) equilibrium for symmetric auctions, e.g., see \citep{fghk2025,krishna2009auction,Milgrom_2004}.

\begin{definition}[$\varepsilon$-approximate symmetric Bayes-Nash equilibrium of
the FPA]\label{def:BNE} Let $\varepsilon
\geq 0$. A symmetric strategy profile $\beta$
is an \emph{(interim) $\varepsilon$-approximate symmetric (pure) Bayes-Nash equilibrium
($\varepsilon$-BNE)} if, for any value $v \in
\support{F}$, it holds that:
\begin{equation}
\label{eq:MBNE-def-condition-full}
u(\beta(v),\beta;v) \geq u(b,\beta;v) - \varepsilon \qquad \text{for all}\;\; b\in B.
\end{equation}
We will refer to a $0$-BNE as an \emph{exact} symmetric BNE.
\end{definition}

As it is standard in the literature, we will be interested in equilibria that satisfy two desirable properties, namely, \emph{no overbidding} and \emph{monotonicity}.
The former property stipulates that bidders do not submit bids larger than their values, as doing that would always result in non-positive utility. The latter property considers strategies that are non-decreasing in the bidder's value. From now on, we will only consider symmetric equilibria which are both non-overbidding and monotone.

\begin{remark}
For positive results like the ones presented in our work, restricting attention to symmetric, non-overbidding, monotone equilibria is without loss of generality, as it only makes the results stronger. 
\end{remark}

\subsection{Distribution Representation} 
\label{sec:represent}

Since $F$ is a continuous distribution, we need to discuss how to represent it in a way that is suitable for computational purposes. We will consider three different representations, all of which are standard in the literature of related problems with continuous functions in their inputs, e.g., see \citep{deligkas2021two,hollender2025envy,fghlp2021_sicomp}. \medskip

\noindent \textbf{Black-box (Oracle) Model:} In this model, we assume that the algorithm has \emph{oracle} or \emph{query access} to the distribution, i.e., the distribution will be accessed via \emph{queries} that input a value $x\in V=[0,1]$ and output the cdf $F(x)$ of the distribution.  The efficiency of an algorithm will then be measured by the (worst-case) number of such queries that it performs. Furthermore, for our positive results, we will make sure that the number of any additional arithmetic operations needed is polynomial (in the relevant parameters). Such a model is of particular interest for applications where the designer cannot have perfect knowledge of the distributional value prior, but can perform some kind of ``market analysis'' and ``estimate'' the bidder population quantiles. \medskip

\noindent \textbf{White-box Implicit Model:}  In this model, we assume that the algorithm has access to some ``meaningful'' computational device, e.g., a Turing machine or an appropriate arithmetic circuit (see, e.g., \citep[Appendix B]{deligkas2021two}). The algorithm will then access the cdf of the distribution for a given value $x \in V=[0,1]$ via ``feeding'' $x$ to this device and obtaining $F(x)$. Since the input to the algorithm includes a description of this device rather than an abstract cdf oracle, this model is typically referred to as the \emph{white-box} model. 

Concretely, we will assume that the input to the algorithm is a \emph{well-behaved arithmetic circuit} $x\mapsto F(x)$, as defined in \citep[Section 3.1.3]{fearnley2021complexity}. These are circuits with gates performing the arithmetic operations $\{+,-,\times,\max,\min,>\}$ and rational constants, with the additional restriction that, for any computational path that leads to an output, there is at most a logarithmic number of multiplication gates. This restriction is required to ensure that the circuit does not ``internally'' generate numbers of magnitude (inversely) doubly-exponential (in the size of the input and the description of the circuit), preventing its efficient evaluation. Furthermore, we will assume that the input to the problem also includes a parameter $L$, such that the distribution is $L$-Lipschitz continuous.\footnote{That is, $\card{F(x)-F(y)}\leq L\cdot \card{x-y}$ for all $x,y\in[0,1]$. See also~\cref{def:lipschitz}.}
\medskip

\noindent \textbf{Explicit Model:} In this model, we assume that the distribution can be represented in a succinct form, and thus can be given directly as input to the algorithm. More precisely, following \citep[Sec.~6]{fghlp2021_sicomp}, we will assume that the cdf of the distribution is piecewise polynomial, and hence it will be provided in the input in terms of the coefficients of the polynomial in each piece.  Formally, the distribution is given in the input by:
    \begin{itemize}\label{page:represent}
    \item a partition of $V$ into $k$ subintervals $\sset{[v_{j-1},v_j]}_{j=1}^k$, given by their (rational) end-points $$0=v_0<v_1<v_2<\ldots<v_k=1;$$
    \item for each $j\in[k]$, a vector of $(d+1)$ rationals $(a_{j,0},a_{j,1},\dots,a_{j,d})$; this induces the semantics that the value of the cdf $F$ within the $j$-th interval is given by the $d$-degree polynomial
    \begin{equation}
        \label{eq:piecewise-poly-represent-cdf}
        F_j(x)\coloneqq \sum_{\ell=0}^d a_{j,\ell} x^\ell,\qquad\forall x\in [v_{j-1},v_{j}].
    \end{equation}
    \end{itemize}
    Note that, in order for this representation to be valid, it needs to respect the boundary conditions
$$
F_1(0)=0, \qquad F_k(1)=1,\quad\text{and}\quad F_j(v_j)=F_{j+1}(v_j)\quad\text{for all}\;\; j\in[k-1],
$$
as well as (weak) monotonicity;
it is not hard to see that all these are properties checkable in polynomial time. 
\medskip

\noindent Our algorithms for our positive results in \cref{sec:cfpa} and \cref{sec:ccfpa} will be applicable to both the black-box and the white-box models, by ensuring that, besides a polynomial number of oracle calls, they perform at most polynomially-many arithmetic operations, and the magnitude of associated numerical values is polynomially bounded.  

\subsection{Bidding Space}\label{sec:prelims-bids}
So far, we have not imposed any conditions on the bidding space $B$. Following the related literature on the computational complexity of the problem (see, e.g., \citep{fghlp2021_sicomp,fghk24,fghk2025,chen2023complexity}), for our main results in \cref{sec:cfpa} we will assume that the bidding space $B$ is \emph{discrete}. More precisely, $B=\{b_1,b_2,\ldots,b_m\}$ is a \emph{finite} subset of $[0,1]$, which is given explicitly as part of the input, by listing its elements\footnote{For all bidding spaces that we consider, we will have $b_1=0$, which can naturally be interpreted as the choice of the bidder to abstain from the auction.} $0=b_1 < b_2< \dots < b_m < 1$.  Again following the literature, we will refer to this variant as the \emph{Continuous First-Price Auction (CFPA)}.\footnote{The term ``continuous'' here refers to the value space which is the whole interval $[0,1]$, as opposed to discrete values spaces which have also been studied in the literature, e.g., in \citep{fghk24,fghk2025}.} 

As a result, every (monotone) bidding strategy $\beta: V\map B$ in the CFPA is a piecewise-constant function and will be represented by specifying \emph{jump points} $0=s_0\leq s_1\leq \ldots \leq s_m=1$ such that $\beta(v)=b_j$ when $v \in (s_{j-1},s_j]$ (and $\beta(s_0)=\beta(0)=b_1=0)$.
Notice that the ordering of the jump points is not strict; this means that some of them might coincide, resulting in some bids of $B$ not being ``used'' by $\beta$.
This will be the canonical representation of strategies, and hence of equilibria, throughout the paper; see \cref{fig:discretebidding} for a pictorial view. We also note that previous works on the problem \citep{fghlp2021_sicomp,chen2023complexity,fghk2025} use the same representation for the output of the auction.

\begin{figure}[t]\begin{center}
    \begin{tikzpicture}[scale=2.0]
        \draw[->] (0,0) -- (5,0);
        \draw[->] (0,0) -- (0,3);
        \node[below] at (0,0) {$s_0 = 0$};
        \node[below] at (0.5,0) {$s_1$};
        \node[below] at (0.9,0) {$s_2$};
        \node[below] at (2.5,-0.2) {$s_{j}$};
        \node[below] at (2.5,0) {$s_{j-1}$};
        \node[below] at (3.6,0) {$s_{m-1}$};
        \node[below] at (4.2,0) {$s_{m}=1$};
        \node[left] at (0,0) {$b_1 = 0$};
        \node[left] at (0,1.5) {$b_{j}$};
        \node[left] at (0,0.4) {$b_{2}$};
        \node[left] at (0,0.6) {$b_{3}$};
        \node[left] at (0,1) {$b_{j-1}$};
        \node[left] at (0,1.8) {$b_{j+1}$};
        \node[left] at (0,2.25) {$b_{m-1}$};
        \node[left] at (0,2.5) {$b_{m}$};
        \coordinate (b1l) at (0,0.0) {};
        \coordinate (b1r) at (0.5,0.0) {};
        \coordinate (b2l) at (2,1) {};
        \coordinate (b2r) at (2.5,1) {};
        \coordinate (b3l) at (2,1.8) {};
        \coordinate (b3r) at (2.5,1.8) {};
        \coordinate (b4l) at (3.6,2.5) {};
        \coordinate (b4r) at (4.2,2.5) {};
        \draw[black!50!gray,ultra thick] (0.5,0.4) -- (0.9,0.4);
        \draw[black!50!gray,ultra thick] (b1l) -- (b1r);
        \draw[black!50!gray,ultra thick] (2,1) -- (b2r);
        \draw[black!50!gray,ultra thick] (2.5,1.8) -- (b3r);
        \draw[black!50!gray,ultra thick] (b4l) -- (b4r);
        \draw[black!50!gray,ultra thick] (2.5,1.8) -- (2.7,1.8);
        \draw[black!50!gray,ultra thick] (0.9,0.6) -- (1.2,0.6);
        \draw[black!50!gray,ultra thick] (3.3,2.25) -- (3.6,2.25);
        \draw[dotted] (0,0.4) -- (0.5,0.4);
        \draw[dotted] (0.5,0) -- (0.5,0.4);
        \draw[dotted] (0.9,0) -- (0.9,0.6);
        \draw[dotted] (2.5,1.8) -- (2.5,0.0);
        \draw[dotted,ultra thick] (1.2,0.6) -- (1.4,0.6);
        \draw[dotted] (0,1) -- (1.8,1);
        \draw[dotted, ultra thick] (b2l) -- (1.8,1);
        \draw[dotted, ultra thick] (2.7,1.8) -- (2.9,1.8);
        \draw[dotted, ultra thick] (3.1,2.25) -- (3.3,2.25);
        \draw[dotted] (b4r) -- (4.2,0.0);
        \draw[dotted] (0,2.5) -- (b4l);
        \draw[dotted] (0,1.8) -- (2.5,1.8);
        \draw[dotted] (3.6,0) -- (3.6,2.5);
        \draw[dotted] (0,0.6) -- (0.9,0.6);
        \draw[dotted] (0,2.25) -- (3.1,2.25);
        \draw[black!50!gray,thick,fill=black!50!gray] (3.6,2.25) circle (1pt);
        \draw[black!50!gray,thick,fill=black!50!gray] (0,0) circle (1pt);
        \draw[black!50!gray,thick,fill=black!50!gray] (b1r) circle (1pt);
        \draw[black!50!gray,thick,fill=black!50!gray] (0.9,0.4) circle (1pt);
        \draw[black!50!gray,thick,fill=black!50!gray] (b2r) circle (1pt);
        \draw[black!50!gray,thick,fill=black!50!gray] (b4r) circle (1pt);
    \end{tikzpicture}
\end{center}
\caption{A monotone bidding strategy $\beta$, succinctly represented
by its jump points $0=s_0\leq s_1 \leq \cdots \leq s_m=1$. \label{fig:discretebidding}}
\end{figure} 

In \cref{sec:ccfpa}, we will also present results for the case when the bidding space is \emph{continuous}, i.e., the entirety of the unit interval, that is, $B=[0,1]$. Following \citep{fghk2025}, we will refer to this variant as the \emph{Continuous First-Price Auction with Continuous Bids (CCFPA)}. 

\section{Continuous First-Price Auction (CFPA)}
\label{sec:cfpa}

In this section, we present our main result for the CFPA, namely a polynomial-time algorithm for computing (approximate) equilibria of the auction. Formally, the main theorem of the section is the following:

\begin{theorem}\label{thm:CDFPA-main-theorem}
There exists an algorithm (see \cref{alg:cdfpa-outer}), which, for any $\varepsilon >0$, and any instance of the CFPA with $L$-Lipschitz value distribution $F$ and minimum bid distance\footnote{That is, if $B=\{b_1,b_2,\ldots,b_m\}$ is the bidding space, then $\alpha := \min_{i\in [m]} (b_{i+1}-b_i)$. Here, we let $b_{m+1} := 1$, i.e., $\alpha := \min\{b_2-b_1, b_3-b_2, \dots, 1-b_m\}$.} $\alpha$, computes an $\varepsilon$-BNE
\begin{itemize}
    \item using $\poly(n, \allowbreak m, \allowbreak \log L, \allowbreak \log(1/\eps), \allowbreak \log(1/\alpha))$ queries to $F$ and operations (additions, multiplications, and divisions) in the \emph{black-box model},
    \item in time polynomial in $\log(1/\varepsilon)$ and the size of the input, in the \emph{white-box model}.
\end{itemize}
\end{theorem}

\paragraph{High-level Overview.} Before we provide the proof of the theorem, we offer some useful discussion. As we mentioned in the introduction, in the context of finding equilibria, the literature in economics has not in general managed to provide analytic solutions for the CFPA. In particular, this means that we cannot attempt to obtain approximate equilibria of the auction, by approximating such (continuous) analytic solutions via (discrete) algorithms.\footnote{This in itself turns out to be a technically challenging task, which we investigate in \cref{sec:ccfpa} for the CCFPA, for which such analytic solutions are known.}

Our algorithm is based on the following crucial observation: if we know the utility achieved by the bidders at equilibrium at the highest point in the value space, i.e., at $v=1$, then we can efficiently compute the equilibrium, i.e., the set of jump points that describe the equilibrium bidding strategy.\footnote{For example, for the uniform distribution, the utility at equilibrium at point $v=1$ is always $1/n$, where $n$ is the number of bidders; this is consistent with the method of \citet{chwe1989discrete} for describing an equilibrium of the CFPA when the value distribution is uniform, and the special case of equidistant bids.}
Additionally, when we run our method for obtaining an equilibrium from a utility value $U$ at $v=1$ that does not correspond to the utility value at an equilibrium, it still outputs a set of jump points. Of course, these jump points do not correspond to an equilibrium in that case, but we can still extract useful information from them. We show that our method for computing a set of jump points $\vec{s}=(s_0, \dots,s_m)$ from a utility value has the following two crucial properties: (a) on input $U=0$, the method outputs $s$ with $s_0 = 0$, and (b) on input $U=1$, it outputs $s$ with $s_0 > 0$. 

This suggests the following approach: perform binary search on the value $U$ in order to find a value $U^*$ where the output $s$ is such that $s_0$ is ``barely above'' $0$. Indeed, we are able to show that if we can find such a value $U^*$ then we obtain an approximate equilibrium. The most challenging part of the analysis is showing that our method is continuous, in the sense that the jump points that it outputs depend continuously on the input value $U$. This is crucial in order to argue that binary search can locate a solution.

\subsection{Useful Notation and Preprocessing}

We present the following useful lemma introducing some key notation regarding the computation of the utilities.
A closely related lemma appears in \citep{chwe1989discrete}; we include our version for completeness and to harmonize the notation with the rest of the paper. Its proof can be found in~\cref{appendix:lem:utility-deltafunc}.

\begin{lemma}[\cite{chwe1989discrete}]
    \label{lem:utility-deltafunc}
    Consider a symmetric strategy $\beta$, given by its canonical jump point representation $0=s_0\leq s_1 \leq \ldots \leq s_m=1$. Assuming that all other bidders use strategy $\beta$, the utility of a bidder when bidding $b_j$ at value $v \in [0,1]$ can be expressed as:
    $$u(b_j,\beta; v) = (v - b_j) \deltafunc(s_{j-1},s_j)$$
    where 
    \begin{equation*}
    \deltafunc(x,y) := \frac{1}{n}\sum_{i=0}^{n-1} F(x)^{n-1-i}F(y)^i. 
    \end{equation*}
    In particular, the probability of any given bidder winning the auction when submitting bid $b_j$ is given by $\deltafunc(s_{j-1},s_j)$.
\end{lemma}

The next lemma shows that, when studying the efficient computation of $\eps$-BNE, it suffices to consider the case where the cdf $F$ is $\eps$-strongly increasing. Thus, from now on we will assume that the cdf satisfies this assumption. Formally, we have the following definition.

\begin{definition}[\emph{$\gamma$-strongly increasing function}]\label{def:strongly-increasing}
Let $f: [a,b] \rightarrow \mathbb{R}$ be a continuous function. Then, given $\gamma >0$, $f$ is $\gamma$-strongly increasing if $f(x) \geq f(y) + \gamma(x-y)$ for any $x,y \in [a,b]$ with $x \geq y$.
\end{definition}

\begin{lemma}[Strong-Monotonicity Preprocessing]\label{lem:wlog-strongly-increasing}
Assume that we have an efficient\footnote{Recall that, in the black-box model, an efficient algorithm is one that performs at most polynomially many (in all the relevant parameters, as mentioned in \cref{thm:CDFPA-main-theorem}) queries and arithmetic operations. In the white-box model, an efficient algorithm is one that runs in polynomial time in the size of the input, where $\eps$ is given in binary representation.} algorithm that for any $\eps > 0$ computes an $\eps$-BNE whenever the cdf $F$ is $\eps$-strongly increasing. Then, we can obtain an efficient algorithm that given any $\eps > 0$ computes an $\eps$-BNE, without any strong-monotonicity assumption.
\end{lemma}

\begin{proof}
Let $\eps > 0$ be given.
We can set $\delta = \eps/3n$ and construct a modified cdf $F'$ by simply letting $F'(x) := \delta x + (1-\delta)F(x)$. Note that $F'$ is a valid cdf, it is $L$-Lipschitz, and it is $\delta$-strongly increasing. Furthermore, we can easily (and efficiently) simulate access to $F'$ in the black-box model, or construct a well-behaved arithmetic circuit for $F'$ given a well-behaved arithmetic circuit for $F$ in the white-box model.

By assumption, we can efficiently compute a $\delta$-BNE $\beta$ of the auction with the $\delta$-strongly increasing cdf $F'$. We will show that $\beta$ is an $\eps$-BNE of the auction with the original cdf $F$. Let $u$ and $u'$ denote the utility functions under $F$ and $F'$ respectively. Since $\beta$ is a $\delta$-BNE of the auction with cdf $F'$, we have that for all $v \in [0,1]$ and $i \in [m]$
$$u'(\beta(v);v) \geq u'(b_i;v) - \delta.$$
Thus, since $\delta \leq \eps/3$, in order to show that $\beta$ is an $\eps$-BNE of the auction with cdf $F$, it suffices to show that $|u'(b_i;v) - u(b_i;v)| \leq \eps/3$ for all $v \in [0,1]$ and $i \in [m]$.

Let $(s_0, \dots, s_m)$ be the jump point representation of $\beta$. By \cref{lem:utility-deltafunc} we can write 
\[u'(b_i;v) = (v-b_i)\phi(F'(s_{i-1}),F'(s_i))
\qquad\text{and}\qquad 
u(b_i;v) = (v-b_i)\phi(F(s_{i-1}),F(s_i)),\] 
where $\phi(x,y) = (1/n) \sum_{j=0}^{n-1} x^j y^{n-1-j}$. Since $|v-b_i| \leq 1$, it thus suffices to show that 
\begin{equation}\label{eq:lemma3-3}
|\phi(F'(s_{i-1}),F'(s_i)) - \phi(F(s_{i-1}),F(s_i))| \leq \eps/3.
\end{equation}
To see this, first notice that $|F'(s_i) - F(s_i)| = |\delta s_i - \delta F(s_i)| = \delta |s_i - F(s_i)| \leq \delta = \eps/3n$
and similarly $|F'(s_{i-1}) - F(s_{i-1})| \leq \eps/3n$. Inequality~\eqref{eq:lemma3-3} then follows from the fact that $\phi$ is $n$-Lipschitz, by \cref{lem:lipschitz-deltafunc}. This completes the proof.
\end{proof}

\subsection{Description and Running Time of the Algorithm}\label{sec:algorithms-description}

In this section, we present the algorithm that establishes \cref{thm:CDFPA-main-theorem}, see \cref{alg:cdfpa-outer}.  
Recall that, by \cref{lem:wlog-strongly-increasing}, we may assume without loss of generality that the cdf $F$ is $\eps$-strongly increasing. Furthermore, recall that the bidding space is $B = \{b_1, \dots, b_m\}$ with $0=b_1 < b_2< \dots < b_m < 1$, and with $b_{m+1} := 1$ used for notational convenience.
\begin{algorithm}[t]
\caption{Efficient computation of an $\varepsilon$-BNE in a symmetric CFPA}\label{alg:cdfpa-outer}

\KwIn{$\varepsilon > 0$, $n$, bids $B = \{b_1, \dots, b_m\}$, $L$-Lipschitz and $\eps$-strongly increasing cdf $F$}
\KwOut{Strategy profile $\vec{s} = (s_0, \dots, s_m)$ constituting a symmetric $\varepsilon$-BNE}

$\alpha \gets \min_{i \in [m]} (b_{i+1}-b_i)$\;
$\delta \gets \left(\frac{\eps^{5n} \alpha^{3n}}{100n^3L^4}\right)^m$\;
$U^\ell \gets 0$, $U^r \gets 1$\;
$\vec{s}^r \gets \compstrat(U^r, \delta)$\;

\While{$|U^r - U^\ell| > \delta$}{
    $U \gets (U^\ell + U^r)/2$\;
    $\vec{s} \gets \compstrat(U,\delta)$\;
    \uIf{$s_0 = 0$}{
        $U^\ell \gets U$\;
    }
    \Else{
        $U^r \gets U$\;
        $\vec{s}^r \gets s$\;
    }
}
\Return{$\vec{s}^* := (0, s_1^r, \dots, s_m^r)$}\;

\end{algorithm}
\cref{alg:cdfpa-outer} uses the \compstrat subroutine, which is presented in \cref{alg:cdfpa-inner}. This subroutine has access to all the inputs of the main algorithm and it receives two additional inputs: a utility value $U$ and a precision parameter $\delta$. 

\begin{algorithm}[t]
\caption{\compstrat subroutine}\label{alg:cdfpa-inner}

\KwIn{$U \in [0,1]$, $\delta > 0$}
\KwOut{A vector $\vec{s} = (s_0, \dots, s_m)$ with $0 \leq s_0 \leq \dots \leq s_m = 1$}

$s_m \gets 1$\;
$U_m \gets U$\;

\For{$i=m$ \KwTo $1$}{
    \uIf{$(s_i-b_i)\deltafunc(s_i,s_i) \leq U_i$}{
        $s_{i-1} \gets s_i$\;
        $U_{i-1} \gets U_i$\;
    }
    \uElseIf{$(s_i-b_i)\deltafunc(b_i,s_i) \geq U_i$}{
        $s_{i-1} \gets b_i$\;
        $U_{i-1} \gets 0$\;
    }
    \Else{
        Use $\log(nL/\delta)$ steps of binary search to find $s_{i-1} \in (b_i,s_i)$ satisfying: $$|(s_i-b_i) \deltafunc(s_{i-1},s_i) - U_i| \leq \delta$$
        $U_{i-1} \gets (s_{i-1}-b_i) \deltafunc(s_{i-1},s_i)$\;
    }
}

\Return{$\vec{s} = (s_0, \dots, s_m)$}\;

\end{algorithm}

Before we proceed, we briefly discuss an integral part of \cref{alg:cdfpa-inner}, the step which uses binary search. First of all, note that we necessarily have $s_i > b_i$, since the algorithm never sets the value $s_i$ below $b_{i+1}$. Observe that if the algorithm reaches the ``else'' branch, then the equation $(s_i-b_i) \deltafunc(s_{i-1},s_i) = U_i$ is guaranteed to have a solution; recall that $\deltafunc(x,y) = (1/n) \sum_{j=0}^{n-1} F(x)^j F(y)^{n-1-j}$, see \cref{lem:utility-deltafunc}.  Indeed, since we did not enter the ``if'' branch, we have $(s_i-b_i) \deltafunc(x,s_i) > U_i$ for $x = s_i$. Since we did not enter the ``else if'' branch, we also have $(s_i-b_i) \deltafunc(x,s_i) < U_i$ for $x = b_i$. As a result, there exists some $x \in (b_i, s_i)$ such that $(s_i-b_i) \deltafunc(x,s_i) = U_i$. Furthermore, we can use binary search to find an approximate solution. Since $|s_i-b_i| \leq 1$, the function $x \mapsto (s_i-b_i) \deltafunc(x,s_i)$ is $nL$-Lipschitz; this can be seen by an application of a simple lemma (\cref{lem:lipschitz-deltafunc}) which we state and prove in \cref{app:simple-lemma}. Thus, after at most $\log(nL/\delta)$ steps of binary search we are guaranteed to have found $x \in (b_i,s_i)$ satisfying $|(s_i-b_i) \deltafunc(x,s_i) - U_i| \leq \delta$. The algorithm then sets $s_{i-1} := x$ and proceeds. Note that $\log(nL/\delta)$ is polynomial in the relevant parameters, namely $n,m,\log L, \log(1/\eps),$ and $\log(1/\alpha)$.

The while loop of the main algorithm (\cref{alg:cdfpa-outer}) is executed $\poly(n,m,\log L, \log(1/\eps), \log(1/\alpha))$ times. As a result, the total number of queries to $F$ as well as the number of simple operations is bounded by a polynomial in $n,m,\log L, \log(1/\eps),$ and $\log(1/\alpha)$. In the white-box model, all these quantities are polynomial in the size of the input, so the algorithm runs in polynomial time.

\subsection{Correctness of the Algorithm}\label{sec:algorithm-correctness}

In this section, we show the correctness of the algorithm, i.e., that it computes a symmetric $\varepsilon$-BNE. We begin by showing a straightforward property of the algorithm.

\begin{lemma}\label{lem:cdfpa-algo-endpoints}
Consider any values of $U^\ell$ and $U^r$ in any execution of the while loop of \cref{alg:cdfpa-outer}. Let $\vec{s}^\ell := \compstrat(U^\ell,\delta)$ and $ \vec{s}^r := \compstrat(U^r,\delta)$, as computed by \cref{alg:cdfpa-inner}, for $\delta$ as defined in \cref{alg:cdfpa-outer}. Then,
$$s^\ell_0 = 0
\qquad\text{and}\qquad
s^r_0 > 0.$$
\end{lemma}

\begin{proof}
First, let $\vec{s} = \compstrat(0,\delta)$. We argue that for all $i=m,\ldots,1$, it holds that $s_{i-1} = b_{i}$; this in particular implies that $s_0 = b_1 = 0$. We proceed by induction on $i$.  \medskip

\noindent \textbf{Base Case:} For the first step ($i=m$), since $s_m > b_m$, and since $\deltafunc(s_m,s_m) > 0$ by \cref{lem:utility-deltafunc}, we have that $(s_m-b_m) \deltafunc(s_m,s_m) > U_m = 0$. As a result, $s_{m-1}$ is not computed in the ``if'' branch of the loop. Instead, since $\deltafunc(b_m,s_m) > 0$, $s_{m-1}$ is instead computed in the ``else if'' branch of the for loop, and, consequently, $s_{m-1}=b_m$ and $U_{m-1}=0$.    \medskip 

\noindent \textbf{Induction Step:} Now consider the $i$th step and assume that 
\begin{equation*}s_{i}=b_{i+1} \qquad \text{ and } \qquad U_{i}=0.
\tag{\text{Induction Hypothesis}}
\end{equation*}
\noindent Again, since $s_{i} = b_{i+1} > b_i$, $\deltafunc(s_i,s_i) > 0$, and $\deltafunc(b_i,s_i) > 0$, $s_{i-1}$ is computed in the ``else if'' part of the loop. As a result, we have that $s_{i-1}=b_i$ and $U_{i-1}=0$. \medskip

\noindent Next, let $\vec{s} = \compstrat(1,\delta)$. We argue that for all $i=m,\ldots,1$, it holds that $s_{i-1} = s_{i}$; since $s_m=1$, this in particular implies that $s_0 = 1> 0$. Again, we proceed by induction on $i$. \medskip 

\noindent \textbf{Base Case:} For the first step ($i=m$), since $0< s_m - b_m < 1$, and $0 \leq \deltafunc(s_m,s_m) \leq 1$, we have that $(s_m-b_m) \deltafunc(s_m,s_m) < U_m = 1$. As a result, $s_{m-1}$ is computed in the ``if'' branch of the loop. Consequently, $s_{m-1}=s_m$ and $U_{m-1}=1$.    \medskip 

\noindent \textbf{Induction Step:} Now consider the $i$th step and assume that 
\begin{equation*}s_{i}=s_{i+1} \qquad \text{ and } \qquad U_{i}=1.
\tag{\text{Induction Hypothesis}}
\end{equation*}
\noindent Again, since $0< s_i - b_i < 1$, $U_i=1$, and $\deltafunc(s_i,s_i) \leq 1$, $s_{i-1}$ is computed in the ``if'' part of the loop. As a result, we have that $s_{i-1}=s_i=1 > 0$, and $U_{i-1}=1$. \medskip

\noindent Notice that in \cref{alg:cdfpa-outer}, at the beginning of the while loop, we have that $U^\ell = 0$ and $U^r = 1$. From the discussion above, it follows that at that point, $s_0^\ell =0$ and $s_0^r=1>0$. During any iteration of the while loop, $U^\ell$ either remains unchanged, or we have that $U^\ell = U$, only if $\vec{s} = \compstrat(U,\delta)$ satisfies $s_0 = 0$. From this, it follows that at that point, $\vec{s}^\ell = \compstrat(U^\ell,\delta)$ satisfies $s_0^\ell = 0$. Similarly, $U^r$ either remains unchanged, or we have that $U^r=U$, only if $\vec{s} = \compstrat(U,\delta)$ satisfies $s_0 >0$. It follows that at that point, $\vec{s}^r = \compstrat(U^r,\delta)$ satisfies $s_0^r > 0$. 
\end{proof}

Next, we show that the \compstrat subroutine satisfies a certain form of continuity. This will be crucial to establish that the output of our algorithm is an approximate BNE.

\begin{lemma}\label{lem:cdfpa-algo-lipschitz}
Let $\delta > 0$. Let $U, U' \in [0,1]$ with $|U-U'| \leq \delta$ and $\vec{s}\coloneqq \compstrat(U,\delta)$ and $\vec{s}'\coloneqq \compstrat(U',\delta)$ denote the corresponding outputs of~\cref{alg:cdfpa-inner}. Then it holds that
$$|s_0 - s_0'| \leq (6n^2 L^3 \eps^{-4n} \alpha^{-2n})^{m} \delta.$$
\end{lemma}

\begin{proof}
We use $s_i$ and $U_i$ to denote the variables in the execution of $\compstrat(U,\delta)$, and $s_i'$ and $U_i'$ to denote the corresponding variables in the execution of $\compstrat(U',\delta)$. We will use induction, to show that for all $i =m,m-1, \ldots, 1,0$ we have 
$$|s_i-s_i'| \leq M_i \qquad \text{ and } \qquad |U_i-U_i'| \leq M_i \qquad \text{ where } \ \ 
M_i := (6n^2 L^3 \eps^{-4n} \alpha^{-2n})^{m-i} \delta.$$ 

\noindent \textbf{Base Case:} This clearly holds for $i=m$. \medskip 

\noindent \textbf{Induction Step:} Now consider the $i$th step and assume the bound holds for all previous steps. In particular, we assume that 
\begin{equation*}
    |s_i-s_i'| \leq M_i \qquad \text{ and }  \qquad |U_i-U_i'| \leq M_i, \tag{\text{Induction Hypothesis}}
\end{equation*}
and our goal is to show that $|s_{i-1}-s_{i-1}'| \leq M_{i-1}$ and $|U_{i-1}-U_{i-1}'| \leq M_{i-1}$. \medskip

\noindent Observe that no matter which branch of the ``if'' statement we enter in~\cref{alg:cdfpa-inner}, we necessarily pick $s_{i-1} \in [b_i, s_i]$ satisfying
\begin{equation}\label{eq:lipschitz-invert}
|(s_i-b_i)\deltafunc(s_{i-1},s_i) - \widehat{U}_i| \leq \delta
\end{equation}
where
$$\widehat{U}_i := \max\{(s_i-b_i)\deltafunc(b_i,s_i), \min\{U_i, (s_i-b_i)\deltafunc(s_i,s_i)\}\}$$
and similarly for $s_{i-1}'$. Indeed, if we enter the first branch, then we must have 

$$(s_i-b_i)\deltafunc(s_i,s_i) \leq U_i \ \  \text{ and thus } \ \ \widehat{U}_i = (s_i-b_i)\deltafunc(s_i,s_i).$$
This is because, due to the fact that $s_i \geq b_{i+1} > b_i$, we always have $$(s_i-b_i)\deltafunc(b_i,s_i) < (s_i-b_i)\deltafunc(s_i,s_i).$$
In that branch, $s_{i-1}$ is set to $s_i$ and thus we have $(s_i-b_i)\deltafunc(s_{i-1},s_i) = \widehat{U}_i$, which implies \eqref{eq:lipschitz-invert}. Similarly, if we enter the second branch, it follows that $\widehat{U}_i = (s_i-b_i)\deltafunc(b_i,s_i)$. Since in this branch we set $s_{i-1} = b_i$, we again have $(s_i-b_i)\deltafunc(s_{i-1},s_i) = \widehat{U}_i$, which implies \eqref{eq:lipschitz-invert}. Finally, if we enter the third branch, then it must be that $\widehat{U}_i = U_i$ and it is immediate that \eqref{eq:lipschitz-invert} holds. \medskip 

\noindent By \cref{lem:lipschitz-deltafunc} we know that $\deltafunc$ is $nL$-Lipschitz. Using \cref{lem:lipschitz-operations} it follows that the function $(s_i, U_i) \mapsto \widehat{U}_i$ is $(1+nL)$-Lipschitz, where we used that $|c - b_i| \leq 1$ and $|\deltafunc(c,c')| \leq 1$, for any $c, c' \in [0,1]$. It follows that $|\widehat{U}_i - \widehat{U}_i'| \leq (1+nL) M_i$. \medskip

\noindent Let $g$ be the function that maps $(x,y)$ to the solution $z \in [b_i,x]$ of the equation $(x-b_i)\deltafunc(z,x) = y$. In \cref{lem:lipschitz-invert-algo}, we show that $g$ is $(L \eps^{-4n} (b_{i+1}-b_i)^{-2n})$-Lipschitz. Furthermore, we know that $s_{i-1} = g(s_i,\tilde{U}_i)$ for some $\tilde{U}_i$ satisfying $|\tilde{U}_i - \widehat{U}_i| \leq \delta$. The same also holds for $s_{i-1}'$. Thus, we have that $s_{i-1} = g(s_i,\tilde{U}_i)$ and $s_{i-1}' = g(s_i',\tilde{U}_i')$, where $|s_i-s_i'| \leq M_i$ and
$$|\tilde{U}_i - \tilde{U}_i'| \leq |\widehat{U}_i - \widehat{U}_i'| + 2\delta \leq (1+nL) M_i + 2\delta \leq 2nL M_i.$$
As a result, we obtain that
$$|s_{i-1} - s_{i-1}'| \leq (L \eps^{-4n} (b_{i+1}-b_i)^{-2n}) 2nL M_i \leq 2n L^2 \eps^{-4n} (b_{i+1}-b_i)^{-2n} M_i.$$

\noindent Finally, observe that no matter which branch of the ``if'' statement we enter, we necessarily set
$$U_{i-1} := (s_{i-1} - b_i) \deltafunc(s_{i-1},s_i) + \max\{0, U_i - (s_i-b_i) \deltafunc(s_i,s_i)\}$$
and similarly for $U_{i-1}'$. By \cref{lem:lipschitz-deltafunc} we know that $\deltafunc$ is $nL$-Lipschitz. By \cref{lem:lipschitz-operations} it follows that the function $(s_{i-1},s_i) \mapsto (s_{i-1} - b_i) \deltafunc(s_{i-1},s_i)$ is $(1+nL)$-Lipschitz, where again we used that $|c - b_i| \leq 1$ and $|\deltafunc(c,c')| \leq 1$, for any $c,c' \in [0,1]$. Similarly, by \cref{lem:lipschitz-operations} it also follows that the function $(s_i, U_i) \mapsto \max\{0, U_i - (s_i-b_i) \deltafunc(s_i,s_i)\}$ is $(1+1+nL)$-Lipschitz. As a result, the function $(s_{i-1},s_i,U_i) \mapsto U_{i-1}$ is $3nL$-Lipschitz, since $2nL+3 \leq 3nL$. Thus, we obtain that
$$|U_{i-1} - U_{i-1}'| \leq 3nL \cdot 2n L^2 \eps^{-4n} (b_{i+1}-b_i)^{-2n} M_i = 6 n^2 L^3 \eps^{-4n} (b_{i+1}-b_i)^{-2n} M_i.$$
Using the fact that $b_{i+1}-b_i \geq \alpha$, it follows that $|U_{i-1} - U_{i-1}'| \leq M_{i-1}$ and $|s_{i-1} - s_{i-1}'| \leq M_{i-1}$, as desired.
\end{proof}

Before we prove the correctness of the algorithm, we state a final useful proposition, which identifies a set of sufficient conditions for the existence of an approximate equilibrium.

\begin{proposition}[Approximate Equilibrium Conditions]
    \label{prop:characterize-approx-equilibrium}
    Let $\gamma \geq 0$.
    Consider a symmetric strategy profile where all bidders use a strategy $\beta$, given by its canonical representation $0 = s_0 \leq s_1 \leq \dots \leq s_m = 1$. Furthermore, assume that there exist $U_0, U_1, \dots, U_m \in [0,1]$ such that for any $i \in [m]$ we have:
    \begin{conditions}[topsep=5pt,itemsep=0.5ex,partopsep=1ex,parsep=1ex]
        \item if $s_{i-1} < s_i$ then $\lvert u(b_i;s_i) - U_i \rvert \leq \gamma$ and $\lvert u(b_i;s_{i-1}) - U_{i-1} \rvert \leq \gamma$ \label{eq-cond:1}
        \item if $s_{i-1} = s_i$ then $U_{i-1} = U_i$ and $u(b_i; s_i) \leq U_i + \gamma$ \label{eq-cond:2}
        \item $s_{i-1} \geq b_i$ \label{eq-cond:3}
    \end{conditions}
    Then, this strategy profile constitutes a $2\gamma m$-BNE.
\end{proposition}

\begin{proof}
   To simplify notation, it is useful to distinguish between bids that are ``actually'' used by the strategy $\beta$, and those that are not.
   In this proof, the former, i.e., bids $b_i$ with $s_{i-1} < s_i$, will be called \emph{effective bids}.
   To begin with, consider two consecutive effective bids, $b_j$ and $b_{j'}$. This means that we have $s_{j-1} < s_j = s_{j'-1} < s_{j'}$ and $s_i = s_j$ for any $i$ satisfying $j < i < j'$.
   By \cref{eq-cond:2}, it must be that $U_j = U_{i} = U_{j'-1}$ for any $i$ with $j < i < j'$.
   Then, by \cref{eq-cond:1}, we get that $\left| u(b_j;s_j)-U_j \right| \leq \gamma \text{ and } \left| u(b_{j'};s_j)-U_j \right| \leq \gamma$.
   From this, we can derive that, for any consecutive effective bids $b_j$ and $b_{j'}$ it must be that:
   \begin{equation}\label{eq:bound-on-consecutive}
    \left| u(b_j;s_j)-u(b_{j'};s_j) \right| \leq 2\gamma
   \end{equation}
   Furthermore, for any non-effective bid $b_i$ lying between consecutive effective bids $b_j$ and $b_{j'}$, \cref{eq-cond:2} yields
   \begin{equation}\label{eq:bound-on-non-effective}
    u(b_i;s_i) - u(b_{j'};s_i) \leq 2\gamma \quad \text{ and } \quad u(b_i;s_i) - u(b_{j};s_i) \leq 2\gamma
   \end{equation}
   
   To show that $\beta$ is an approximate equilibrium, consider some value $v\in (s_{\ell-1},s_\ell]$, for some $\ell \in [m]$ with $s_{\ell-1} < s_\ell$.
   By definition of the strategies, we know that $\beta(v) = b_\ell$.
   Consider any deviation to some $b_i$; we will show that this cannot increase the bidder's utility by more than $2\gamma m$.
   First, consider the case where $i<\ell$.
   Define $j_+$ to be the smallest index that is greater than or equal to $i$ for which $b_{j_+}$ is an effective bid, i.e.
   $$
   j_+ \coloneq \min \{ j: i \leq j \wedge s_{j-1}<s_j\}
   $$
   Notice that, if $b_i$ is an effective bid, then $b_i=b_{j_+}$.
   Let $b_{r_1},b_{r_2},\ldots,b_{r_p}$ be the effective bids from $b_{j_+}$ to $b_\ell$ (inclusive).
   Here, we note that it could be that $p=1$, in which case $b_{j_+}=b_\ell$, meaning that $b_\ell$ is precisely the first effective bid that is greater than or equal to $b_i$.
   Then, we have:
   \begin{equation}\label{eq:diff-in-utilities}
    u(b_i;v)-u(b_\ell;v) = u(b_i;v) - u(b_{j_+};v) + \sum_{k=1}^{p-1} \left[u(b_{r_k};v) - u(b_{r_{k+1}};v)\right]
   \end{equation}
   Consider the difference $u(b_i;v) - u(b_{j_+};v)$.
   As we mentioned earlier, we can see that if $b_i$ is an effective bid, then $b_i=b_{j_+}$ and this difference is equal to $0$. If $b_i$ is not an effective bid, then $b_{j_+}$ is the closest larger bid that is effective and by \eqref{eq:bound-on-non-effective} we obtain that $u(b_i;s_i) - u(b_{j_+};s_i) \leq 2\gamma$.
   Using the definition of the utilities, we can rewrite the inequality as follows:
   \begin{align*}
   & \qquad\quad (s_i-b_i) \Delta(s_{i-1},s_i) - (s_i-b_{j_+}) \Delta(s_{j_+ -1},s_{j_+}) \leq 2\gamma \\
   & \Longrightarrow \quad s_i (\Delta(s_{i-1},s_i) - \Delta(s_{j_+ -1},s_{j_+})) - b_i \Delta(s_{i-1},s_i) + b_{j_+} \Delta(s_{j_+ -1},s_{j_+}) \leq 2\gamma \\
   & \Longrightarrow \quad v (\Delta(s_{i-1},s_i) - \Delta(s_{j_+ -1},s_{j_+})) - b_i \Delta(s_{i-1},s_i) + b_{j_+} \Delta(s_{j_+ -1},s_{j_+}) \leq 2\gamma \\
   & \Longrightarrow \quad u(b_i;v) - u(b_{j_+};v) \leq 2\gamma 
   \end{align*}
   where in the second step we used the fact that $s_i \leq v$ and $\Delta(s_{i-1},s_i) - \Delta(s_{j_+ -1},s_{j_+}) \leq 0$, since the winning probability when bidding $b_{j_+}$ is at least as large as when bidding $b_i$ (because $b_{j_+} \geq b_i$).

   Similarly, for any two consecutive effective bids $b_{r_k}$ and $b_{r_{k+1}}$, by \eqref{eq:bound-on-consecutive} we have $u(b_{r_k};s_{r_k}) - u(b_{r_{k+1}};s_{r_k}) \leq 2\gamma$. Using the same argument as above, we obtain that $u(b_{r_k};v) - u(b_{r_{k+1}};v) \leq 2\gamma$. For this, we use the fact that $s_{r_k} \leq v$ and the probability of winning when bidding $b_{r_{k+1}}$ is at least as large as when bidding $b_{r_k}$.

   Thus, looking at the sum of \eqref{eq:diff-in-utilities}, we see that each summand is bounded by $2\gamma$. 
   Hence, we obtain
   \begin{equation*}
    u(b_i;v)-u(b_\ell;v) \leq 2 \gamma p \leq 2 \gamma m
   \end{equation*}
   The case for $i>\ell$ follows by a symmetric argument.
   Therefore, we have shown that, for any value $v$, any deviation from the strategy $\beta$ cannot increase a bidder's utility by more than $2\gamma m$, meaning that $\beta$ is a $2\gamma m$-BNE.
\end{proof}

We finally prove the correctness of the algorithm. 

\begin{theorem}
\cref{alg:cdfpa-outer} outputs a symmetric $\eps$-BNE.
\end{theorem}

\begin{proof}
At the end of the algorithm, we have that $|U^r - U^\ell| \leq \delta$ (for $\delta$ fixed as in the preamble of~\cref{alg:cdfpa-outer}). As a result, by \cref{lem:cdfpa-algo-lipschitz} we must have
$$|s^\ell_0 - s^r_0| \leq (6n^2 L^3 \eps^{-4n} \alpha^{-2n})^{m} \delta$$
where $s^\ell = \compstrat(U^\ell,\delta)$ and $s^r = \compstrat(U^r,\delta)$. Furthermore, by \cref{lem:cdfpa-algo-endpoints}, we have $s^\ell_0 = 0$ and $s^r_0 > 0$. Thus, we have
$$0 < s^r_0 \leq (6n^2 L^3 \eps^{-4n} \alpha^{-2n})^{m} \delta.$$
In particular, by the choice of $\delta$, we have $s^r_0 \leq \alpha/2$ and $s^r_0 \leq \eps/4mnL$. Furthermore, we also have $\delta \leq \eps/4m$.\medskip

\noindent We first argue that $(s^r_i-b_i)\deltafunc(b_i,s^r_i) < U_i$ for all $i \in [m]$. Indeed, assume that there exists some $i \in [m]$ with $(s^r_i-b_i)\deltafunc(b_i,s^r_i) \geq U_i$. Since $s^r_i > b_i$ and $\deltafunc(\cdot, s^r_i)$ is strictly increasing, it follows in particular that $(s^r_i-b_i)\deltafunc(s^r_i,s^r_i) > U_i$. As a result, in the $i$th step of the for loop in subroutine \compstrat, the ``else if'' branch is taken and thus we have $s^r_{i-1} = b_i$ and $U_{i-1} = 0$. From this it then follows that $s^r_0 = 0$, a contradiction.\medskip

\noindent Let $\beta$ denote the bidding strategy represented by $\vec{s} := (0, s^r_1, \dots, s^r_m)$. Note that this is the output of the algorithm. We will show that it is an $\eps$-BNE. For any $i \in \{2,\dots,m\}$ with $s^r_{i-1} = s^r_{i}$, we must have $U_{i-1} = U_{i}$ and $(s^r_i-b_i)\deltafunc(s^r_i,s^r_i) \leq U_i$. For any $i \in \{2,\dots,m\}$ with $s^r_{i-1} < s^r_i$, by the previous paragraph it must be that $|(s^r_i-b_i)\deltafunc(s^r_{i-1},s^r_i) - U_i| \leq \delta \leq \eps/2m$ and $(s^r_{i-1}-b_i)\deltafunc(s^r_{i-1},s^r_i) = U_{i-1}$.\medskip

\noindent We have $s^r_0 \leq \alpha/2 < b_2 \leq s_1^r$, i.e., $s^r_0 < s_1^r$, and thus $|(s^r_1-b_1)\deltafunc(s^r_0,s^r_1) - U_1| \leq \delta$. By \cref{lem:lipschitz-deltafunc} we know that $\deltafunc$ is $nL$-Lipschitz. As a result, we have $|(s^r_1-b_1)\deltafunc(0,s^r_1) - U_1| \leq \delta + nLs^r_0 \leq \eps/2m$. We have thus shown that $\vec{s} = (0, s^r_1, \dots, s^r_m)$ satisfies the conditions of \cref{prop:characterize-approx-equilibrium} with $\gamma = \eps/2m$. As a result, $\vec{s}$ is an $\eps$-BNE.
\end{proof}

\section{Continuous First-Price Auction with Continuous Bids (CCFPA)}\label{sec:ccfpa}

We now turn to the study of the CCFPA, where both the value space and the bidding space are continuous. From seminal works in the economics literature~\citep{riley1981optimal,MW82,Vickrey1962}, we know\footnote{For a detailed derivation, the reader is referred to the textbook of~\citet[Sec.~2.3]{krishna2009auction}} that the following (symmetric) bidding strategy constitutes an (exact) BNE (\cref{def:BNE}) of the CCFPA:
\begin{equation}
    \label{eq:CCFPA-IID-BNE-beta}
    \beta^*(v) \coloneqq v-\int_{\underline{v}}^v \frac{F^{n-1}(t)}{F^{n-1}(v)} \,\mathrm{d} t \qquad\qquad\text{for all}\;\; v \in [\underline{v},1].
\end{equation}
In the following, we will refer to $\beta^*$ as the \emph{canonical equilibrium
strategy} of the CCFPA. Observe that the bidding function $\beta^*$ is
(absolutely) continuous and nondecreasing in $[\underline{v},1]$,
\emph{strictly} increasing in the support $\support{F}$ of the value distribution $F$, and it also satisfies no-overbidding (that is, $\beta^*(v)\leq v$ for all
$v\in[\underline{v},1]$).\footnote{For a more detailed discussion of the
analytic properties of $\beta^*$ see, e.g., the textbooks of
\citet[Sec.~3.1.1]{Menezes2005} or \citet[Sec.~2.3]{krishna2009auction},
and the paper of~\citet[Sec.~6.2]{fghk2025}.} Purely for technical
convenience,\footnote{\label{footnote:support}This is inconsequential from a Bayesian equilibrium analysis
perspective (see~\cref{def:BNE}), since a bidding strategy need only be defined in the value
distribution's support.} we can extend $\beta^*$ to be defined on the entirety
of $V=[0,1]$, by setting $\beta^*(v)=v$ for $v\in[0,\underline{v}]$.  Note that this
extension preserves the properties of continuity, monotonicity, and
no-overbidding of $\beta^*$. 

\subsection{Black-Box Input}\label{sec:black-box}

We begin our study of the computational complexity of the CCFPA, by considering the more general, black-box input for the item value distributions (see~\cref{sec:represent}).
We fully determine the query complexity of finding (approximate) equilibria in first-price auctions with continuous bidding spaces, by providing matching upper and lower bounds, in \cref{th:black-box-upper} and \cref{th:black-box-lower}, respectively. We first present our upper bound in the following theorem:

\begin{theorem}[Black-box Upper Bound] \label{th:black-box-upper}
An $\varepsilon$-BNE of a CCFPA can be computed via $O(1/\varepsilon)$ oracle calls to the cdf of its value distribution. More precisely, given $\varepsilon>0$, there exists a continuous, no-overbidding and monotone, (symmetric) bidding strategy $\beta:V\map B$ such that:
\begin{itemize}
    \item $\beta$ is an $\varepsilon$-approximate equilibrium of the CCFPA  (\cref{def:BNE});
    \item $\beta$ is $\varepsilon$-near the canonical equilibrium strategy $\beta^*$~\eqref{eq:CCFPA-IID-BNE-beta}, i.e., $\cards{\beta^*(x)-\beta(x)}\leq \varepsilon$ for all $x\in V$; and
    \item given any $x\in V$, $\beta(x)$ can be computed by performing (at most) $\frac{1}{\varepsilon}+1$ queries to the cdf $F$ of the CCFPA's distribution of item values and $O(\frac{n}{\varepsilon})$ arithmetic operations (additions, multiplications, and divisions).
\end{itemize}
\end{theorem}

\begin{proof}
Fix an $\varepsilon>0$. 
    Let $K\coloneqq
    \left\lceil\frac{1}{\varepsilon}\right\rceil$ and $\hat{\varepsilon}\coloneqq
    \frac{1}{K}$. Note that $0<\hat{\varepsilon}\leq\varepsilon$. We partition the value space $V=[0,1]$
    in $\hat{\varepsilon}$-length subintervals, by defining the discretization points
    $a_j\coloneqq j\cdot \hat{\varepsilon}$, for $j=0,1,\dots, K$. 
    
    For any $x\in V$, we define the integer
    \begin{equation}
        \label{eq:CCFPA-FPTAS-black-box-k-def}
        {k_x} \coloneqq \left\lfloor\frac{x}{\hat{\varepsilon}}\right\rfloor=\left\lfloor x K\right\rfloor,
    \end{equation}
    and the function $g_x:[0,x]\map[0,1]$ by
    \begin{equation}
        \label{eq:CCFPA-black-box-g-x-def}
        g_x(t)\coloneqq
        \begin{cases}
            1, &\text{if}\;\; 0\leq x \leq \underline{v},\\
            1-\frac{F^{n-1}(t)}{F^{n-1}(x)}, &\text{if}\;\; \underline{v}< x \leq 1.
        \end{cases}
    \end{equation}
    Note that $g_x$ is well-defined, since $F(x)>0$ for all $x>\underline{v}$.
    Furthermore, it is straightforward to see that every function $g_x$ is
    continuous and nonincreasing, since $F^{n-1}(t)$ is nondecreasing for
    $t\in V$. For the case of $x>\underline{v}$, in particular, it holds that
    $g_x(0)=g_x(\underline{v})=1$ and $g_x(x)=0$; that is, $g_x$ is constantly
    equal to $1$ within interval $[0,\underline{v}]$, and then continuously and
    monotonically decreases to $0$ within interval $[\underline{v},x]$. As a
    matter of fact, since the cdf of the value distribution $F$ is strictly
    increasing in its support, for $x>\underline{v}$ we can also derive that
    $g_x$ is strictly decreasing in $\support{F}$. We summarize these
    observations in the following:
    \begin{property}
        \label{prop:monotonicity-g-x-wrt-t}
        For all $x>\underline{v}$, function $g_x:[0,x]\map[0,1]$ defined
        in~\eqref{eq:CCFPA-black-box-g-x-def} is:  continuous; onto; \emph{strictly}
        decreasing in $\support{F}$; constant in $[0,x]\setminus{\support{F}}$. For $x\leq
        \underline{v}$, function $g_x$ is constantly equal to $1$.
    \end{property}

    Using functions $g_x$, the canonical equilibrium
    strategy~\eqref{eq:CCFPA-IID-BNE-beta} of our CCFPA can be now simply
    expressed as 
    $$
    \beta^*(x) = \int_{0}^x g_x(t)\, \mathrm{d} t
    \qquad\text{for all}\;\; x\in V.
    $$
    Therefore, by defining the lower and upper Riemann sums of function $g_x$,
    with respect to the discretization $0=a_0<a_1<a_2<\dots<a_{{k_x}}\leq x$, i.e.
    \begin{align}
        L(x)\coloneqq & \sum_{j=1}^{{k_x}+1} [\hat{a}_j(x)-\hat{a}_{j-1}(x)]\cdot g_x\left(\hat{a}_{j}(x)\right) \label{eq:CCFPA-black-box-lower-riemann}\\
        U(x)\coloneqq & \sum_{j=1}^{{k_x}+1} [\hat{a}_j(x)-\hat{a}_{j-1}(x)]\cdot g_x\left(\hat{a}_{j-1}(x)\right)  \label{eq:CCFPA-black-box-upper-riemann}\\
        \intertext{where}
        \hat{a}_j(x)\coloneqq &
        \begin{cases}
            a_j=j\hat{\varepsilon}, &\text{if}\;\; j=0,1,\dots,{k_x},\\
            x, &\text{if}\;\; j={k_x}+1,
        \end{cases}\label{eq:CCFPA-black-box-x-specific-discretization}
    \end{align}
    we can immediately deduce that 
    \begin{equation}
        \label{eq:CCFPA-black-box-beta-riemann-bounds}
        L(x)\leq \beta^*(x) \leq U(x) \qquad\text{for all}\;\; x\in V.
    \end{equation}
    
    In the remainder of the proof, we will show that function $U$
    from~\eqref{eq:CCFPA-black-box-upper-riemann} is a valid bidding strategy,
    that satisfies the desired properties of the statement of our theorem. That
    is, we will be using   
    \begin{equation}
        \label{eq:CCFPA-black-box-FPTAS-tilde-beta-def}
        \beta(x)\coloneqq U(x) 
        \underset{\eqref{eq:CCFPA-black-box-x-specific-discretization}}{\overset{\eqref{eq:CCFPA-black-box-upper-riemann}}{=}}
        (x-a_{k_x})g_x(a_{k_x})+\hat{\varepsilon}\sum_{j=0}^{{k_x}-1} g_x(a_j),
    \end{equation}
    for all $x\in V=[0,1]$.

    \paragraph{No-overbidding.} From~\cref{prop:monotonicity-g-x-wrt-t}, we know
    that $0\leq g_x(t)\leq 1$ for all $x\in V$, $t\in[0,x]$. Thus, since
    $x-a_{k_x}=x-k_x\hat{\varepsilon}\overset{\eqref{eq:CCFPA-FPTAS-black-box-k-def}}{\geq}
    0$ all $x\in[0,1]$, from~\eqref{eq:CCFPA-black-box-FPTAS-tilde-beta-def} we
    can upper-bound $\beta$ by:
    $
    \beta(x)\leq (x-k_x\hat{\varepsilon})\cdot 1 + \hat{\varepsilon}\cdot k_x \cdot 1 =x
    $, for any $x\in V$.

    \paragraph{Approximation to $\beta^*$.} Because of~\eqref{eq:CCFPA-black-box-beta-riemann-bounds} we have that
    \begin{align*}
        \label{eq:CCFPA-FFPTAS-black-box-beta-tilde-approximation}
        0 &\leq \beta(x) - \beta^*(x) \\
        &\leq U(x)-L(x)\\
        &= 
        \sum_{j=1}^{{k_x}+1} [\hat{a}_j(x)-\hat{a}_{j-1}(x)]\left[ g_x\left(\hat{a}_{j-1}(x)\right) - g_x\left(\hat{a}_{j}(x)\right) \right], &&\text{due to \eqref{eq:CCFPA-black-box-lower-riemann} and \eqref{eq:CCFPA-black-box-upper-riemann},}\\
        &\leq \sum_{j=1}^{{k_x}+1} \hat{\varepsilon} \left[ g_x\left(\hat{a}_{j-1}(x)\right) - g_x\left(\hat{a}_{j}(x)\right) \right], &&\text{due to \eqref{eq:CCFPA-black-box-x-specific-discretization},}\\
        &= \hat{\varepsilon} \left[ g_x\left(\hat{a}_{0}(x)\right) - g_x\left(\hat{a}_{{k_x}+1}(x)\right) \right]\\
        &= \hat{\varepsilon} \left[ g_x\left(0\right) - g_x\left(x\right) \right], &&\text{due to \eqref{eq:CCFPA-black-box-x-specific-discretization},}\\
        &\leq \hat{\varepsilon} \left[ 1 - 0 \right]\\
        &= \varepsilon.
    \end{align*}    
    
    \paragraph{Continuity and monotonicity.} We will show that function
    $\beta:V\map B$, as defined above
    by~\eqref{eq:CCFPA-black-box-FPTAS-tilde-beta-def}, is a continuous,
    nondecreasing function (in the entirety of $V$); and \emph{strictly}
    increasing in $\support{F}$. To do that, we
    first state and prove the following property, regarding the continuity and
    monotonicity of $g_x(t)$ with respect to $x$:  
    \begin{property}
        \label{claim:g-x-increasing-wrt-x}
        For any $t\in V$, quantity $g_x(t)$ is continuous and nondecreasing with
        respect to $x\in[t,1]$. Furthermore, if $t>\underline{v}$, quantity $g_x(t)$
        is \emph{strictly} increasing for $x\in[t,1]\inters \support{F}$.
    \end{property}
    \begin{nestedproof}[Proof of~\cref{claim:g-x-increasing-wrt-x}.]
        First, observe that, if $t\leq \underline{v}$ then $F(t)=0$; so, no matter
        whether $x\leq \underline{v}$ or $x> \underline{v}$, both branches
        in~\eqref{eq:CCFPA-black-box-g-x-def} give us that $g_x(t)=1$ for all $x\geq
        t$, trivially satisfying~\cref{claim:g-x-increasing-wrt-x}.
        
        For the remaining case of $\underline{v}<t\leq x$, we know that we are in
        the second branch of~\eqref{eq:CCFPA-black-box-g-x-def}; that is,
        $g_x(t)=1-\frac{F^{n-1}(t)}{F^{n-1}(x)}$. Also, it must be that $F(t)>0$,
        therefore from the (continuity and) monotonicity of the cdf $F$, we can
        immediately deduce that $\frac{F^{n-1}(t)}{F^{n-1}(x)}$ is continuous with
        respect to $x$, and in particular: constant outside the support of $F$; and
        strictly decreasing in the support of $F$.
    \end{nestedproof}
    
    Suppose that the support of the value distribution $F$ starts in the $j^*$-th
    subinterval of our $\hat{\varepsilon}$-partition of V. That is, let
    $j^*\in [K]$ be such that $\underline{v}\in [a_{j^*-1},a_{j^*})$. Then, for
    any $x\in[0,a_{j^*})$ it must be that $a_j\leq \underline{v}$ for all
    $j=0,1,\dots,k_x$, and thus (utilizing~\cref{prop:monotonicity-g-x-wrt-t})
    expression~\eqref{eq:CCFPA-black-box-FPTAS-tilde-beta-def} gives us in this
    case that $\beta(x) = (x-k_x\hat{\varepsilon})\cdot 1
    +\hat{\varepsilon} \sum_{j=0}^{k_x-1} 1 = x$. Therefore, within
    $[0,a_{j^*})$, $\beta$ is the identity function, and is indeed
    continuous and strictly increasing. In particular, note that it must be that
    $\beta(\underline{v})= \underline{v}$.
    
    Now we move to the remaining case of $x\in [a_{j^*},1]$. In this case, first
    observe that it is $a_{k_x}>\underline{v}$. Furthermore,
    function $k_x$ (defined in~\eqref{eq:CCFPA-FPTAS-black-box-k-def}) remains
    constant within any interval $[a_j,a_{j+1})$, with
    $a_{k_x}=a_j$. Therefore, utilizing
    \cref{claim:g-x-increasing-wrt-x}, we can deduce that (i) quantity
    $x-a_{k_x}=x-a_j$ is continuous, nonnegative, and strictly
    increasing with respect to $x\in[a_j,a_{j+1})$; (ii) quantity
    $g_x(a_{k_x})=g_x(a_j)$ is continuous, nonnegative, and
    nondecreasing with respect to $x\in[a_j,a_{j+1})$; in particular, its
    monotonicity is \emph{strict} for $x\in[a_j,a_{j+1})\inters \support{F}$; and (iii)
    quantity $\hat{\varepsilon}\sum_{j=0}^{{k_x}-1}
    g_x(a_j)$ is continuous and nondecreasing with
    respect to $x\in[a_j,a_{j+1})$. From the above, by inspecting
    expression~\eqref{eq:CCFPA-black-box-FPTAS-tilde-beta-def}, we deduce that
    function $\beta$ is continuous and nondecreasing in all intervals
    $[a_j,a_{j+1})$, for $j=j^*,j^*+1,\dots,K-1$; and its monotonicity is strict
    within $\support{F}$.  
    
    Finally, to conclude the proof of the desired continuity and monotonicity
    properties for the bidding function $\beta$, it remains to be shown
    that $\beta$ is \emph{left}-continuous on points
    $a_1,a_2,\dots,a_K$; that is, on the jump points of function $k_x$
    (see~\eqref{eq:CCFPA-FPTAS-black-box-k-def}). Consider such a $j\in[K]$.
    Then, $k_{a_j}=j$ and so
    from~\eqref{eq:CCFPA-black-box-FPTAS-tilde-beta-def} we get that 
    $$
    \beta(a_j)= (a_j-a_j)g_{a_j}(a_j) + \hat{\varepsilon}\sum_{\ell=0}^{j-1} g_{a_j}(a_\ell)=\hat{\varepsilon}\sum_{\ell=0}^{j-1} g_{a_j}(a_\ell).
    $$
    On the other hand, for any $x\in(a_{j-1},a_j)$ it is $k_{x}=j-1$, and so, by
    making use of the continuity given by~\cref{claim:g-x-increasing-wrt-x},
    this time~\eqref{eq:CCFPA-black-box-FPTAS-tilde-beta-def} gives that 
    \begin{align*}
        \lim_{x\to a_j^{-}}\beta(x) 
        &=  \lim_{x\to a_j^{-}}\left(x-a_{j-1}\right)\cdot \lim_{x\to a_j^{-}}g_x\left(a_{j-1}\right)+\hat{\varepsilon}\sum_{\ell=0}^{j-2} \lim_{x\to a_j^{-}} g_x\left(a_{\ell}\right)\\
        &=\left(a_j-a_{j-1}\right) \cdot g_{a_j}(a_{j-1}) + \hat{\varepsilon}\sum_{\ell=0}^{j-2} g_{a_j}(a_\ell)\\
        &=\hat{\varepsilon} \cdot  g_{a_j}(a_{j-1}) + \hat{\varepsilon}\sum_{\ell=0}^{j-2} g_{a_j}(a_\ell)
        =\hat{\varepsilon}\sum_{\ell=0}^{j-1} g_{a_j}(a_\ell)\\
        &= \beta(a_j).
    \end{align*} 
    
    \paragraph{Approximate equilibrium.} 
    For the remainder of this proof, let $u(b;v)$ denote the expected (interim)
    utility of a player, when she has true value $v$ and submits bid $b$, while
    all other players bid according to the same bidding function
    $\beta$. To establish that $\beta$ is an $\varepsilon$-BNE,
    we need to show that, for any $v\in \support{F}$ and $b\in B=[0,1]$, it holds that 
    \begin{equation}
        \label{eq:CCFPA-FPTAS-approximate-BNE-to-show}
        u(b;v)-u(\beta(v);v) \leq \varepsilon.
    \end{equation}
    Since the value distribution $F$ is continuous, and $\beta$ is
    strictly increasing in the support of $F$, we can express these bidder
    utilities as:
    \begin{equation}
        \label{eq:bounds-utility-general-CCFPA-FPTAS}
        u(b;v) = \probs{\beta(Y)\leq b} (v-b) \qquad\text{for all}\;\; v\in \support{F},\; b\in B,
    \end{equation}
    where $Y$ is the random variable of the maximum of $n-1$ i.i.d.\ draws from
    the value distribution $F$.
    
    Now recall that, earlier in our proof when studying the monotonicity of
    $\beta$, we showed that
    $\beta(\underline{v})=\underline{v}$. Thus,
    $\beta(\support{F})\subseteq[\underline{v},1]$. In particular, this means that
    $\beta(Y)\geq \underline{v}$. So, for any $b<\underline{v}$ it is
    $\probs{\beta(Y)\leq b}=0$, and so from
    \eqref{eq:bounds-utility-general-CCFPA-FPTAS} we get that $u(b;v)=0$. This
    means that~\eqref{eq:CCFPA-FPTAS-approximate-BNE-to-show} is trivially
    satisfied for $b<\underline{v}$. 
    
    Therefore, it remains to be shown
    that~\eqref{eq:CCFPA-FPTAS-approximate-BNE-to-show} holds for all $b\in[
    \underline{v},1]$ as well. Indeed, fix such a $b\geq \underline{v}$ and let
    $z$ denote the value up to which a bidder cannot beat bid $b$; that is, we
    define
    \begin{equation}
        \label{eq:beta-tilde-inverse-def}
        z \coloneq \sup\ssets{v\in \support{F}\fwhs{\beta(v)\leq b}}.
    \end{equation}
    Note that the supremum above is well-defined, since 
    $\beta(\underline{v})=\underline{v}\leq b$. Also, due to the
    continuity of the function $\beta$, it must be that $\beta(z)
    \leq b$. Therefore, for any $v\in \support{F}$ we have that 
    \begin{align}
        u(b;v) &= \probs{Y\leq z} (v-b), &&\text{from~\eqref{eq:bounds-utility-general-CCFPA-FPTAS} and~\eqref{eq:beta-tilde-inverse-def}},\notag\\
        &= F^{n-1}(z) (v-b)\notag\\
        &\leq F^{n-1}(z) (v-\beta(z)), &&\text{since $\beta(z) \leq b$,}\notag\\
        &\leq F^{n-1}(z) (v-\beta^*(z)), &&\text{since $\beta^*\leq\beta$,}\notag\\
        &= F^{n-1}(z)(v-z)+\int_{\underline{v}}^z F^{n-1}(t) \,\mathrm{d} t, &&\text{from~\eqref{eq:CCFPA-IID-BNE-beta},}\label{eq:FPTAS-black-box-approx-BNE-helper-1}
    \end{align}
    and 
    \begin{align}
        u(\beta(v);v) &= \probs{\beta(Y)\leq \beta(v)} (v-\beta(v)), &&\text{due to~\eqref{eq:bounds-utility-general-CCFPA-FPTAS}},\notag\\
        &= \probs{Y\leq v} (v-\beta(v)), &&\text{since $\beta$ is strictly increasing in $\support{F}\ni v$},\notag\\
        &= F^{n-1}(v) (v-\beta(v))\notag\\
        &\geq F^{n-1}(v) (v-\beta^*(v)-\varepsilon), &&\text{since $\beta \leq \beta^*+\varepsilon$,}\notag\\
        &=\int_{\underline{v}}^v F^{n-1}(t) \,\mathrm{d} t -F^{n-1}(v)\varepsilon, &&\text{from~\eqref{eq:CCFPA-IID-BNE-beta}.}\label{eq:FPTAS-black-box-approx-BNE-helper-2}
    \end{align}
    
    Combining \eqref{eq:FPTAS-black-box-approx-BNE-helper-2} and~\eqref{eq:FPTAS-black-box-approx-BNE-helper-1} we finally get that
    \begin{align*}
        u(b;v) - u(\beta(v);v) 
        &\leq F^{n-1}(z)(v-z)+\int_{v}^z F^{n-1}(t) \,\mathrm{d} t +F^{n-1}(v)\varepsilon\\
        &\leq F^{n-1}(z)(v-z)+\int_{v}^z F^{n-1}(z) \,\mathrm{d} t +F^{n-1}(1)\varepsilon, &&\text{since $F$ is nondecreasing},\\
        &= F^{n-1}(z)(v-z)+(z-v) F^{n-1}(z) +\varepsilon,\\
        &=\varepsilon,
    \end{align*}
    establishing ~\eqref{eq:CCFPA-FPTAS-approximate-BNE-to-show}.

    \paragraph{Efficient computation.}\label{page:ccfpa-black-box-upper-running-time} 
    To be able to compute the value of
    $\beta(x)$, given an $x\in V$ in the input, we perform the
    following operations:
    \begin{itemize}
        \item First, as a pre-processing step, independent of the input $x$, we
        perform $K-1=\lceil 1/\varepsilon \rceil-1$ cdf queries, in order to
        find the values of $F(a_1),F(a_2),\dots,F(a_{K-1})$; note that, we
        already know that $F(a_0)=F(0)=0$ and $F(a_K)=F(1)=1$. Using these, we
        compute the values
        $\sset{F^{n-1}(a_j)}_{j=0,1,\dots,K}$ and store them in a look-up table,
        for easy future reference. This computation requires $O(Kn)=O(n/\varepsilon)$ arithmetic operations (multiplications).
        \item When we receive $x\in V$, first we perform just one additional
        cdf query, to get the value of $F(x)$.
        \begin{itemize}
            \item If $F(x)=0$, this means that $x\leq \underline{v}$, and thus we
            can immediately return the value $\beta(x)=x$.
            \item If $F(x)>0$, then $x>\underline{v}$, and we know that we are at
            the second branch of expression~\eqref{eq:CCFPA-black-box-g-x-def}.
            Therefore, from~\eqref{eq:CCFPA-black-box-FPTAS-tilde-beta-def}, the
            value of $\beta(x)$ can be expressed as 
            \begin{align}
                \beta(x) &= (x-a_{k_x})g_x(a_{k_x})+\hat{\varepsilon}\sum_{j=0}^{{k_x}-1} g_x(a_j) \notag\\
                &= (x-a_{k_x})\left[1-\frac{F^{n-1}(a_{k_x})}{F^{n-1}(x)}\right]+\hat{\varepsilon}\sum_{j=0}^{{k_x}-1} \left[1-\frac{F^{n-1}(a_j)}{F^{n-1}(x)}\right]\notag\\
                &= x-a_{k_x}-(x-a_{k_x})\frac{F^{n-1}(a_{k_x})}{F^{n-1}(x)}+k_x\hat{\varepsilon}-\frac{\hat{\varepsilon}}{F^{n-1}(x)}\sum_{j=0}^{{k_x}-1} F^{n-1}(a_j)\notag\\
                &= x - \frac{1}{F^{n-1}(x)}\left[\hat{\varepsilon}\sum_{j=0}^{{k_x}-1} F^{n-1}(a_j)+(x-k_x\hat{\varepsilon})F^{n-1}(a_{k_x}) \right]. \label{eq:computation-formula-approx-BNE-CCFPA}
            \end{align}
            This quantity can be computed \emph{exactly}, by
            simply computing the integer $k_x$
            from~\eqref{eq:CCFPA-FPTAS-black-box-k-def}, and then fetching the first
            $k_x+1$ values of $\ssets{F^{n-1}(a_j)}_j$ from our pre-processing look-up table,
            together with the value of $F^{n-1}(x)$ that we queried in this step. Therefore, the computation in~\eqref{eq:computation-formula-approx-BNE-CCFPA} requires at most $O(k_x+n)=O(n/\varepsilon)$ arithmetic operations (additions, multiplications, and a division).
        \end{itemize}
    \end{itemize}
\end{proof}

Next, we state our query complexity lower bound, which establishes that the corresponding upper bound of \cref{th:black-box-upper} is (asymptotically) the best possible:
\begin{theorem}[Black-box Lower Bound]
\label{th:black-box-lower}
    Any  algorithm\footnote{Note that our lower bound construction is purely information-theoretic, and does not rely on any assumptions for the algorithm's running time, or a specific computational model.} that computes a (symmetric, no-overbidding, continuous, and increasing) $\varepsilon$-approximate BNE of a CCFPA, needs to make at least $\varOmega(1/\varepsilon)$ cdf queries to its value distribution.
\end{theorem}

\begin{proof}
    Fix a symmetric CCFPA setting with $n=2$ bidders and a continuous, strictly increasing value distribution (with cdf) $F$ over $V=[0,1]$. Let $\mathcal{A}$ be an arbitrary deterministic algorithm that has only oracle access to $F$. For our lower bound construction, let's fix that, whenever $\mathcal{A}$ submits to us a query $F(x)$, for $x\in V$, we will return $F(x)=x$ to it. In that way, even if $\mathcal A$ is an adaptive algorithm, after making $k$-many oracle calls to us, on points $x_1,x_2,\dots,x_k\in V$, it will receive the values
    \begin{equation}
        \label{eq:CCFPA-blackbox-lower-bound-pretend-uniform}
        F(x_1)=x_1,F(x_2)=x_2,\dots,F(x_k)=x_k.
    \end{equation}
    In other words, we let $\mathcal{A}$ ``believe'' that the underlying prior distribution is a uniform one; and we need to make sure that our actual prior $F$ is compatible with these returned queries. In the following, let $F_1(x)=x$, for all $x\in V$, denote the uniform distribution over the unit-interval $[0,1]$; obviously, $F_1$ satisfies this compatibility requirement~\eqref{eq:CCFPA-blackbox-lower-bound-pretend-uniform}.
    
    By the pigeonhole principle, after making $k$ queries, there has to exist a $\frac{1}{3\cdot(k+1)}$-length sub-interval of $[2/3,1]$ which is not ``hit'' by any of $\mathcal{A}$'s queries. That is, there exist $v_1,v_2\in[\frac{2}{3},1]$ such that 
    \begin{equation}
        \label{eq:CCFPA-blackbox-lower-bound-pigeonhole}
        v_2-v_1=\delta \coloneqq\frac{1}{3\cdot(k+1)} \qquad\text{and}\qquad x_1,x_2,\dots,x_k\notin (v_1,v_2).
    \end{equation}
    Fix a parameter $0<\xi<\delta$ (to be taken arbitrarily close to zero, at the end of our proof) and define the probability distribution $F_2$ with cdf 
    \begin{equation*}
        F_2(x) \coloneqq 
        \begin{cases}
            x, & x\in [0,v_1]\union [v_2,1],\\
            v_1 + \frac{\xi}{\delta-\xi}(x-v_1), & x\in[v_1,v_2-\xi],\\
            v_1+\xi+\frac{\delta-\xi}{\xi}(x-v_2+\xi), & x\in[v_2-\xi,v_2].
        \end{cases}
    \end{equation*}
    Note that $F_2$ is a continuous, strictly increasing cdf over $[0,1]$, which is also compatible with our query constraints~\eqref{eq:CCFPA-blackbox-lower-bound-pretend-uniform}. In particular, observe that $F_1(x)=F_2(x)$ for all $x\in[0,1]\setminus(v_1,v_2)$. A pictorial representation can be seen in~\cref{fig:black-box-lower-bound}.
   \begin{figure}[t]
    \centering
    \begin{tikzpicture}
        
        % --- variables ---
        \pgfmathsetmacro{\vone}{0.75}   % v_1 (in [2/3,1])
        \pgfmathsetmacro{\deltanum}{0.12}  % delta = v_2 - v_1
        \pgfmathsetmacro{\xinum}{0.03}     % xi in (0,delta)
        \pgfmathsetmacro{\vtwo}{\vone+\deltanum} % v_2
        \pgfmathsetmacro{\vk}{\vtwo-\xinum}      % kink point v_2 - xi
        \pgfmathsetmacro{\voneplusxi}{\vone+\xinum}
        \pgfmathsetmacro{\vmid}{(\vone+\vtwo)/2}
        \pgfmathsetmacro{\axisstart}{0.63}
        \pgfmathsetmacro{\xmeas}{\axisstart+0.015} % x-position for vertical \xi brace (inside plot)
        \pgfmathsetmacro{\ymeas}{\axisstart+0.015} % y-position for horizontal \xi brace (inside plot)
        \pgfmathsetmacro{\ydelta}{0.92} % y-position for the \delta length marker (near top, below label)
        \pgfmathsetmacro{\deltatick}{0.007} % half tick length (in axis units)
        \pgfmathsetmacro{\deltaoffset}{0.004} % inset from v_1 and v_2 for marker endpoints
        \pgfmathsetmacro{\vonedelta}{\vone+\deltaoffset}
        \pgfmathsetmacro{\vtwodelta}{\vtwo-\deltaoffset}
        
        % Optional: zoom window (set xmin=0,xmax=1 for "full view". Currently, we start from 2/3)
        \begin{axis}[
            width=0.70\linewidth,
            height=0.48\linewidth,
            xmin=\axisstart, xmax=1,
            ymin=\axisstart, ymax=1,
            axis lines=left,
            xlabel={$x$}, ylabel={$F(x)$},
            xtick={2/3,\vone,\vk,\vtwo,1},
            xticklabels={$2/3$,$v_1$,$v_2-\xi$,$v_2$,$1$},
            ytick={2/3,\vone,\voneplusxi,\vtwo,1},
            yticklabels={$2/3$,$v_1$,$v_1+\xi$,$v_2$,$1$},
            legend style={draw=black, fill=black!3, at={(0.02,0.97)}, anchor=north west},
            legend cell align={left}, % The command for legend alignment
            ]
            
            % Shade the “unqueried” interval (v1,v2)
            \addplot[draw=none, fill=black!12, forget plot]
            coordinates {(\vone,\axisstart) (\vtwo,\axisstart) (\vtwo,1) (\vone,1)} \closedcycle;

            % Highlight that the slice (v_1,v_2) has length \delta
            \draw[thick]
                (axis cs:\vonedelta,\ydelta) -- (axis cs:\vtwodelta,\ydelta)
                node[midway, yshift=7pt] {$\delta$};
            \draw[thick] (axis cs:\vonedelta,\ydelta-\deltatick) -- (axis cs:\vonedelta,\ydelta+\deltatick);
            \draw[thick] (axis cs:\vtwodelta,\ydelta-\deltatick) -- (axis cs:\vtwodelta,\ydelta+\deltatick);
            
            % F1(x)=x (thick outside (v1,v2), thin inside)
            \addplot[thick] coordinates {(\axisstart,\axisstart) (\vone,\vone)};
            \addlegendentry{$F_1$ and $F_2$}
            \addplot[thin, red] coordinates {(\vone,\vone) (\vtwo,\vtwo)};
            \addlegendentry{$F_1$}
            \addplot[thick, forget plot] coordinates {(\vtwo,\vtwo) (1,1)};
            
            % F2: coincides with F1 outside (v1,v2)
            \addplot[thick, dashed, forget plot] coordinates {(\axisstart,\axisstart) (\vone,\vone)};
            \addplot[thick, dashed, forget plot] coordinates {(\vtwo,\vtwo) (1,1)};
            
            % F2 inside (v1,v2): two linear pieces
            \addplot[thin, dashed, blue, forget plot] coordinates {(\vone,\vone) (\vk,\voneplusxi)};
            \addplot[thin, dashed, blue] coordinates {(\vk,\voneplusxi) (\vtwo,\vtwo)};
            % Legend entry for F2
            \addlegendentry{$F_2$}
            
            % Vertical guides at v1, v2-xi, v2, 1
            \addplot[thin, dotted] coordinates {(\vone,\axisstart) (\vone,\vone)};
            \addplot[thin, dotted] coordinates {(\vk,\axisstart) (\vk,\voneplusxi)};
            \addplot[thin, dotted] coordinates {(\vtwo,\axisstart) (\vtwo,\vtwo)};
            \addplot[thin, dotted] coordinates {(1,\axisstart) (1,1)};
            
            % Horizontal guides from y-axis to the marked points
            \addplot[thin, dotted] coordinates {(\axisstart,\vone) (\vone,\vone)};
            \addplot[thin, dotted] coordinates {(\axisstart,\voneplusxi) (\vk,\voneplusxi)};
            \addplot[thin, dotted] coordinates {(\axisstart,\vtwo) (\vtwo,\vtwo)};
            \addplot[thin, dotted, forget plot] coordinates {(\axisstart,1) (1,1)};
            
            % Point markers
            \addplot[only marks, mark=*] coordinates {(\vone,\vone) (\vk,\voneplusxi) (\vtwo,\vtwo)};

            % Emphasize that both gaps are of length \xi
            % Vertical gap: v_1 to v_1+\xi
            \draw[thick, decorate, decoration={brace, amplitude=4pt, mirror}]
                (axis cs:\xmeas,\vone) -- (axis cs:\xmeas,\voneplusxi)
                node[midway, xshift=10pt] {$\xi$};
            % Horizontal gap: v_2-\xi to v_2
            \draw[thick, decorate, decoration={brace, amplitude=4pt}]
                (axis cs:\vk,\ymeas) -- (axis cs:\vtwo,\ymeas)
                node[midway, yshift=10pt] {$\xi$};
            
            % Label for the hidden interval
            \node[anchor=north] at (axis cs:\vmid,0.99)
            {\small unqueried interval};
        \end{axis}
    \end{tikzpicture}
    
    \caption{Black-box lower-bound construction of~\cref{th:black-box-lower}: $F_1$ is uniform; $F_2$ matches $F_1$
    everywhere except on $(v_1,v_2)$, where it is ``flattened then steepened'' (kink at $v_2-\xi$) to remain a valid continuous, strictly increasing cdf.}
    \label{fig:black-box-lower-bound}
\end{figure}

    Let $\beta:V\to B$ be a continuous, no-overbidding (i.e.\
    $\beta(v)\leq v$ for all $v\in V$), and strictly increasing (since $F$ is
    strictly increasing) bidding strategy, returned by
    $\mathcal{A}$ as an $\varepsilon$-BNE, after performing the $k$ queries on
    $x_1,x_2,\dots,x_k$. Then, since $F_1$ and $F_2$ agree on all query
    points, it must be that, if $\mathcal{A}$'s output is correct, $\beta$ must be
    a valid $\varepsilon$-BNE with respect to \emph{both} priors $F_1$ and $F_2$. In particular, this implies that the following inequalities must hold 
    \begin{equation}
        \label{eq:CCFPA-blackbox-lower-bound-equilibrium-conditions-general-1}
    u_1\left(\beta(v_2-\xi);v_1\right) \leq u_1\left(\beta(v_1);v_1\right) + \varepsilon,
    \end{equation}
    \begin{equation}
        \label{eq:CCFPA-blackbox-lower-bound-equilibrium-conditions-general-3}
    u_1\left(v_1/2;v_1\right) \leq u_1(\beta(v_1);v_1) + \varepsilon,
    \end{equation}
    and
    \begin{equation}
        \label{eq:CCFPA-blackbox-lower-bound-equilibrium-conditions-general-2}
    u_2\left(\beta(v_1);v_2-\xi\right) \leq u_2\left(\beta(v_2-\xi);v_2-\xi\right) + \varepsilon,
    \end{equation}
    where, for $i=1,2$, and any $v\in V,b\in B$, $u_i(b;v)$ denotes the expected utility of a player who has true value $v$ and bids $b$, with respect to the other player's value being drawn from distribution $F_i$. 
    Note that, since $\beta$ is strictly increasing (and thus, also invertible), these utilities can be expressed as 
    \begin{equation}
        \label{eq:CCFPA-blackbox-lower-bound-equilibrium-utilities-beta-def}
    u_i(b;v)= (v-b)F_i(\beta^{-1}(b)).
    \end{equation}

    Using~\eqref{eq:CCFPA-blackbox-lower-bound-equilibrium-utilities-beta-def}, inequalities \eqref{eq:CCFPA-blackbox-lower-bound-equilibrium-conditions-general-1} and \eqref{eq:CCFPA-blackbox-lower-bound-equilibrium-conditions-general-2} can be written as
    $$ [v_1-\beta(v_2-\xi)]F_1(v_2-\xi) \leq [v_1-\beta(v_1)] F_1(v_1) + \varepsilon $$
    and
    $$[v_2-\xi-\beta(v_1)] F_2(v_1) \leq [v_2-\xi-\beta(v_2-\xi)] F_2(v_2-\xi) + \varepsilon,$$
    respectively.
    Substituting the cdf values $F_1(v_1)=F_2(v_1)=v_1$, $F_1(v_2-\xi)=v_2-\xi$, and $F_2(v_2-\xi)=v_1+\xi$, and solving with respect to $\beta(v_2-\xi)$, we get that 
    $$ \beta(v_2-\xi) \geq \frac{v_1}{v_2-\xi}\beta(v_1) + v_1\frac{\delta-\xi}{v_2-\xi} - \frac{\varepsilon}{v_2-\xi} $$
    and
    $$\beta(v_2-\xi) \leq \frac{1}{v_1+\xi}\left[v_1\beta(v_1)+(v_2-\xi)\xi+\varepsilon\right]
    < \beta(v_1)+(v_2-\xi)\frac{\xi}{v_1} + \frac{\varepsilon}{v_1},$$
    respectively. For the last (strict) inequality we simply used the fact that $0<v_1<v_1+\xi$.
    Combining the above two inequalities, to eliminate $\beta(v_2-\xi)$, we derive that 
    $$
    \frac{v_1}{v_2-\xi}\beta(v_1) + v_1\frac{\delta-\xi}{v_2-\xi} - \frac{\varepsilon}{v_2-\xi}
    <
    \beta(v_1)+(v_2-\xi)\frac{\xi}{v_1} + \frac{\varepsilon}{v_1},
    $$
    and solving with respect to $\beta(v_1)$, we finally get 
    \begin{align}
        \beta(v_1) &> v_1 - \frac{1}{\delta-\xi}\frac{1}{v_1}\left[(v_2-\xi)^2\xi + (v_1+v_2-\xi)\varepsilon\right]\notag\\
            &> v_1 - \frac{1}{\delta-\xi}\frac{1}{v_1}\left[\xi + (1+v_1)\varepsilon\right],\label{eq:CCFPA-blackbox-lower-bound-equilibrium-conditions-bound-beta_a}
    \end{align}
    where for the last inequality we used the fact that $0<v_2-\xi<1$.

    Next, we turn our attention to equilibrium condition~\eqref{eq:CCFPA-blackbox-lower-bound-equilibrium-conditions-general-3}, which via the utility expression~\eqref{eq:CCFPA-blackbox-lower-bound-equilibrium-utilities-beta-def} and the fact that $F_1$ is the uniform distribution, can be written as
    $$
    \left(v_1-\frac{v_1}{2}\right) \beta^{-1}\left(v_1/2\right) \leq \left(v_1-\beta(v_1)\right)v_1 +\varepsilon.
    $$
    Since $\beta$ is a no-overbidding strategy, it must be that $\beta^{-1}(v_1/2) \geq v_1/2$, and thus the above inequality gives us 
    $$
        \beta(v_1) \leq \frac{3}{4}v_1 + \frac{\varepsilon}{v_1}.
    $$
    Combining this with~\eqref{eq:CCFPA-blackbox-lower-bound-equilibrium-conditions-bound-beta_a} to eliminate $\beta(v_1)$, we derive that
    $$
    v_1 - \frac{1}{\delta-\xi}\frac{1}{v_1}\left[\xi + (1+v_1)\varepsilon\right] < \frac{3}{4}v_1 + \frac{\varepsilon}{v_1}.
    $$
    Rearranging the terms, we get that 
    \begin{align*}
    \frac{1}{\delta-\xi} &> \frac{1}{4}\frac{v_1^2}{\xi+\varepsilon (v_1+1)}- \frac{\varepsilon}{\xi+(v_1+1)\varepsilon}
        > \frac{1}{4}\frac{v_1^2}{\xi+\varepsilon (v_1+1)}- \frac{1}{(v_1+1)},
    \end{align*}
    where for the last inequality we used the fact that $\xi>0$, in the second fraction. Taking $\xi$ arbitrarily close to zero, we can deduce that it must hold that 
    \begin{align*}
    \frac{1}{\delta} 
    \geq \frac{1}{4}\frac{v_1^2}{\varepsilon (v_1+1)}- \frac{1}{(v_1+1)}
    \geq \frac{1}{4}\frac{1}{\varepsilon} \frac{(2/3)^2}{2/3+1} - \frac{1}{2/3+1}
    = \frac{1}{15} \frac{1}{\varepsilon}- \frac{3}{5},
    \end{align*}
    where for the second inequality we used the fact that quantities $\frac{x^2}{x+1}$ and $-\frac{1}{x+1}$ are increasing in $x\in[0,1]$, and the bound $v_1\geq 2/3$. Finally, recalling from \eqref{eq:CCFPA-blackbox-lower-bound-pigeonhole} that $\delta=\frac{1}{3}\frac{1}{k+1}$, we derive the desired lower bound on the number of queries $k$:
    $$
    3(k+1) \geq  \frac{1}{15} \frac{1}{\varepsilon}- \frac{3}{5} \quad \then\quad k \geq  \frac{1}{45} \frac{1}{\varepsilon}- \frac{6}{5}=\varTheta\left(\frac{1}{\varepsilon}\right).
    $$
\end{proof}

\subsubsection{White-Box Implicit Input}\label{sec:white-box}

Upon careful inspection of the proof of \cref{th:black-box-upper}, it can be observed that the number of multiplications performed on any computation path is bounded by $O(\log n)$ (see~\eqref{eq:computation-formula-approx-BNE-CCFPA}). Therefore, the proof effectively establishes a corresponding FPTAS in the white-box model (see~\cref{sec:represent}); this is formally captured by the following result:

\begin{corollary}[White-box FPTAS]
\label{th:white-box-upper}
Consider a CCFPA whose value distribution is given to us implicitly, via a well-behaved arithmetic circuit that computes its cdf. Then, in time polynomial in $\frac{1}{\varepsilon}$, the number of players $n$, and the size of the CCFPA's representation, we can construct a polynomial-time Turing machine that evaluates the $\varepsilon$-approximate equilibrium function $x\mapsto\beta(x)$ described in~\cref{th:black-box-upper}.
\end{corollary}

\subsection{Explicit Representation}
\label{sec:CCFPA-explicit}

We now turn our attention to the explicit representation model (see~\cref{sec:represent}), where the value distribution $F$ is succinctly given to us in the input, in a piecewise polynomial form. Our main result of the section, namely~\cref{th:ccfpa-poly-time-explicit} establishes the polynomial-time computability of \emph{exact} symmetric equilibria in CCFPA. Notice that this result is stronger than \cref{th:white-box-upper}, when instantiated for piecewise polynomial distributions, as the latter is only an approximation scheme, and not an exact algorithm. An additional interesting implication of our particular construction in the proof~\cref{th:ccfpa-poly-time-explicit}, and the exactness of the result, is that players with rational item values end up submitting rational bids, at equilibrium.

\begin{theorem}
\label{th:ccfpa-poly-time-explicit}
    Given a CCFPA with $n$ bidders and a succinctly represented, piecewise polynomial of degree (at most) $d$, value distribution, its canonical equilibrium
    strategy~\eqref{eq:CCFPA-IID-BNE-beta} can be computed efficiently. That is,
    given a value $x\in [\underline{v},1]$, we can (exactly) compute $\beta^*(x)$ in
    time polynomial in $n$, $d$, and the size of the (binary) representation of $x$
    and the auction.\footnote{As a matter of fact, we show something stronger here:
    that we can efficiently compute a \emph{succinct representation} of function
$\beta^*$ as a rational function of degree $(n-1)d+1$ (that is, as the ratio of
   two polynomials with degrees at most $(n-1)d+1$).} 
\end{theorem}

\begin{proof}
We start by showing that, for each interval $[v_{j-1},v_j]$, $j\in[k]$, of our explicit
    representation (see page~\pageref{page:represent}), the corresponding component of the cdf of the
    value distribution, when raised to the power of $n-1$, i.e.\ $F^{n-1}_j$,
    can be expressed as a polynomial of degree (at most) $(n-1)d$ and its
    coefficients can be computed in polynomial time. 
    
    We do this iteratively, by computing the coefficients of the polynomial
    representation of $F^{\powidx}_j$, for all $\powidx=1,2,\dots,n-1$. Let us denote
    this representation by 
    \begin{equation}
        \label{eq:CCFPA-succinct-poly-represent-powers-cdf}
        F^{\powidx}_j(x)=\sum_{\ell=0}^{\powidx d} b_{j,\powidx,\ell}
        x^{\ell},\qquad\text{for}\;\; x\in[v_{j-1},v_j].
    \end{equation}
    \begin{itemize}
        \item For $\powidx=1$, we readily have our initial polynomial
        representation, i.e.\ $b_{j,1,\ell}=a_{j,\ell}$ for
        $\ell=0,1,\dots,d$.
        \item For $\powidx\geq 2$, it is 
        $$F_j^{\powidx}(x) =F_j^{\powidx-1}(x)\cdot F_j(x)= \left(\sum_{\ell=0}^{(\powidx-1)
        d} b_{j,\powidx-1,\ell} x^{\ell}\right) \cdot \left(\sum_{\ell=0}^d
        a_{j,\ell} x^\ell \right)$$ for $x\in [v_{j-1},v_j]$. Therefore, the
        coefficients of the polynomial representation of $F_j^{\powidx}$
        (see~\eqref{eq:CCFPA-succinct-poly-represent-powers-cdf}) can be
        computed, by performing $O(\powidx d^2)$ multiplications and additions, as
        $$
        b_{j,\powidx,\ell} = \sum_{\xi=\max\ssets{0,\ell-(\powidx-1)d}}^{\min\ssets{\ell,d}}b_{j,\powidx-1,\ell-\xi}a_{j,\xi}
        \qquad\qquad\text{for}\;\; \ell=0,1,\dots,\powidx d.
        $$
    \end{itemize}
    At the end of this iterative process, i.e.\ for $\powidx=n-1$, we have computed
    the polynomial representation of $F_j^{n-1}$; that is, parameters
    $\ssets{b_{j,n-1,\ell}}_{\ell=0}^{(n-1)d}$
    in~\eqref{eq:CCFPA-succinct-poly-represent-powers-cdf} for $\powidx=n-1$. This
    requires $O\left(\sum_{\powidx=1}^{n-1} \powidx d^2\right)=O(n^2 d^2)$ basic arithmetic
    operations. It is important to also note that the magnitude of these
    parameters is upper-bounded by 
    $$\max_{\ell=0,1,\dots,(n-1)d} \card{b_{j,n-1,\ell} }
    \leq (d+1) \cdot \max_{\ell} \card{b_{j,n-2,\ell}} \cdot \max_{\ell} \card{a_{j,\ell}}
    \leq \dots
    \leq (d+1)^n\left(\max_{\ell=0,1,\dots,d} \card{a_{j,\ell}}\right)^{n}$$ 
    and, therefore,
    their bit-size representation is polynomial (with respect to the binary
    representation of the auction input). Performing the above computation for
    all intervals in $F$'s representation, we can efficiently compute, with
    $O(kn^2d^2)$ operations, the polynomial representations for all $F_j^{n-1}$,
    $j\in[k]$. 
    
    Based on this, we will next show that the integral $\int_0^x F^{n-1}(t)
    \,\mathrm{d} t$, as a function of $x\in V=[0,1]$, can be succinctly expressed
    as a piecewise polynomial function of degree (at most) $(n-1)d+1$. More
    specifically, we will efficiently compute parameters
    $\ssets{c_{j,\ell}}_{j\in[k],\;\ell=0,1,\dots (n-1)d+1}$ such that
    \begin{equation}
        \label{eq:CCFPA-succinct-poly-represent-integral-numerator-equilibrium}
        \int_0^x F^{n-1}(t) \,\mathrm{d} t = \sum_{\ell=0}^{(n-1)d+1} c_{j,\ell} x^{\ell},\qquad\text{for all}\;\; x\in[v_{j-1},v_{j}],\;\; j\in[k].
    \end{equation}
    We do that iteratively, over the intervals of $F$'s piecewise
    representation:
    \begin{itemize}[leftmargin=*]
        \item For $j=1$, we have that for any $x\in[0,v_1]$ it is:
        \begin{align*}
        \int_0^x F^{n-1}(t) \,\mathrm{d} t &= \int_0^x F_1^{n-1}(t)  \,\mathrm{d} t =
        \int_0^x \sum_{\ell=0}^{(n-1) d} b_{1,n-1,\ell} t^{\ell}\,\mathrm{d} t \\&=
        \left[\sum_{\ell=0}^{(n-1) d} \frac{b_{1,n-1,\ell}}{\ell+1}
        t^{\ell+1}\right]_{t=0}^{t=x}= \sum_{\ell=1}^{(n-1) d+1}
        \frac{b_{1,n-1,\ell-1}}{\ell} x^{\ell},
        \end{align*}
        where for the second equality we
        are using the previously computed representation of $F_1^{n-1}$
        from~\eqref{eq:CCFPA-succinct-poly-represent-powers-cdf}. Therefore, we can
        compute the desired $c_{i,\ell}$ parameters by setting:
        $$
        c_{1,0}=0 \qquad\text{and}\qquad 
        c_{1,\ell}=
        \frac{b_{1,n-1,\ell-1}}{\ell}\quad \text{for}\;\; \ell=1,2,\dots, (n-1)d+1.
        $$
        \item For $j=2,3,\dots,k$ and $x\in[v_{j-1},v_j]$ we can express the integral as
        \begin{align*}
            \int_0^x F^{n-1}(t) \,\mathrm{d} t
            &=\int_{v_{j-1}}^{x} F^{n-1}_{j}(t)\,\mathrm{d} t + \int_{0}^{v_{j-1}} F^{n-1}(t)\,\mathrm{d} t\\
            &=\int_{v_{j-1}}^{x} \sum_{\ell=0}^{(n-1) d} b_{j,n-1,\ell}t^{\ell}\,\mathrm{d} t 
            + \sum_{\ell=0}^{(n-1)d+1} c_{j-1,\ell} v_{j-1}^{\ell}, \ \ \ \ \text{using~\eqref{eq:CCFPA-succinct-poly-represent-powers-cdf} and~\eqref{eq:CCFPA-succinct-poly-represent-integral-numerator-equilibrium}}\\
            &= \left[\sum_{\ell=0}^{(n-1) d} \frac{b_{j,n-1,\ell}}{\ell+1}
            t^{\ell+1}\right]_{t=v_{j-1}}^{t=x} + \sum_{\ell=0}^{(n-1)d+1} c_{j-1,\ell} v_{j-1}^{\ell}\\
            &= \sum_{\ell=1}^{(n-1) d+1} \frac{b_{j,n-1,\ell-1}}{\ell}x^{\ell}
            - \sum_{\ell=1}^{(n-1) d+1} \frac{b_{j,n-1,\ell-1}}{\ell}v_{j-1}^{\ell}
            + \sum_{\ell=0}^{(n-1)d+1} c_{j-1,\ell} v_{j-1}^{\ell}\\
            &= \sum_{\ell=1}^{(n-1) d+1} \frac{b_{j,n-1,\ell-1}}{\ell}x^{\ell}
            + \underbrace{c_{j-1,0}+\sum_{\ell=1}^{(n-1)d+1}\left( c_{j-1,\ell} - \frac{b_{j,n-1,\ell-1}}{\ell} \right) v_{j-1}^{\ell}}_{\text{constant}},
        \end{align*}
        so we need to set 
        \begin{align*}
        c_{j,0}&=c_{j-1,0}+\sum_{\ell=1}^{(n-1)d+1}\left( c_{j-1,\ell} - \frac{b_{j,n-1,\ell-1}}{\ell} \right) v_{j-1}^{\ell}
        \ \text{ and } \\
        c_{j,\ell}&= \frac{b_{j,n-1,\ell-1}}{\ell}\quad\text{for}\;\; \ell=1,2,\dots,(n-1)d+1.
        \end{align*}
    \end{itemize}
    Note that the iterative process above requires $O(kn^2d^2)$ basic arithmetic
    operations, in order to compute all parameters $\ssets{c_{j,\ell}}$.
    
    Finally, we can now argue about the succinct representation of the canonical
    equilibrium bidding function in~\eqref{eq:CCFPA-IID-BNE-beta}, as a
    rational function. For any interval of $F$'s representation $[v_{j-1},v_j]$,
    $j\in[k]$, and any value $x\in [v_{j-1},v_j]$ in this interval, we can use
    representations~\eqref{eq:CCFPA-succinct-poly-represent-powers-cdf}
    and~\eqref{eq:CCFPA-succinct-poly-represent-integral-numerator-equilibrium}
    to write \eqref{eq:CCFPA-IID-BNE-beta} as 
    \begin{align}
        \beta^*(x) 
        &=x-\frac{\int_{0}^x F^{n-1}(t) \,\mathrm{d} t}{F_j^{n-1}(x)}
        = x- \frac{\sum_{\ell=0}^{(n-1)d+1} c_{j,\ell} x^{\ell}}{\sum_{\ell=0}^{(n-1) d} b_{j,n-1,\ell}
        x^{\ell}} \notag\\
        &= \frac{\sum_{\ell=0}^{(n-1) d} b_{j,n-1,\ell}
        x^{\ell+1} -\sum_{\ell=0}^{(n-1)d+1} c_{j,\ell} x^{\ell}}{\sum_{\ell=0}^{(n-1) d} b_{j,n-1,\ell}
        x^{\ell}}
        = \frac{-c_{j,0}+\sum_{\ell=1}^{(n-1)d+1} (b_{j,n-1,\ell-1}-c_{j,\ell}) x^{\ell}}{\sum_{\ell=0}^{(n-1) d} b_{j,n-1,\ell}
        x^{\ell}}, \label{eq:ccfpa-poly-time-explicit-rational-expression}
    \end{align}
    being the ratio of two polynomials, the one on the numerator having degree
    (at most) $(n-1)d+1$ and the one on the denominator having degree $(n-1)d$.
\end{proof}

\section{Conclusion and Future Work}\label{sec:conclusion}

Our work provides the first efficient algorithms for computing equilibria in symmetric first-price auctions, with efficiency measured either by the number of queries to the distribution oracle (in the black-box model) or by running time (in the white-box model). The existence of these algorithms essentially settles the computational complexity of the problem for the case of i.i.d.\ private values, closing an important gap in the literature. This allows future work to focus on \emph{asymmetric} auctions, in particular the case of bidders with (non-identically distributed) independent private values (IPV), for which the computational complexity of equilibrium computation is still open.

Within the realm of symmetric auctions, the only (minor) gap left by our results in this paper is the case of inverse-exponential approximations in the implicit white-box model for the CCFPA, for which our algorithm is an FPTAS; in contrast, in the explicit-input model our corresponding algorithm computes an exact equilibrium in polynomial time. Additionally, our lower bound in \cref{sec:black-box} for the black-box model suggests that any white-box algorithm for $\varepsilon$-approximate equilibria that runs in time polynomial in $\log(1/\varepsilon)$ would need to ``open the box'' provided in the input. We suspect that the existence of such an algorithm would be effectively precluded by an associated PPAD-hardness result. Given the absence of PPAD-hardness results for the more general IPV case, however, obtaining a corresponding result for the symmetric auction setting seems quite challenging.

Finally, one could aim to generalize our algorithms to symmetric settings beyond auctions with i.i.d.\ private values. In particular, it would be meaningful to study the case of \emph{symmetric affiliated private values} of \citet{MW82}; for this case, the best known positive result is a PTAS due to \citet{fghk2025}.

\bibliographystyle{plainnat}
\bibliography{refs}

\newpage 
\appendix
\section*{APPENDIX}

\section[Proof of Lemma~\ref{lem:utility-deltafunc}]{Proof of~\cref{lem:utility-deltafunc}}
\label{appendix:lem:utility-deltafunc}

Let $X_1,X_2,\dots,X_n$ be i.i.d.\ random variables drawn from $F$, representing the values of the bidders. Due to symmetry, it is without loss of generality, enough to show that the probability that the first bidder wins the auction, when she submits $b_j$, is equal to $\deltafunc(s_{j-1},s_j)$. The first bidder gets the item if and only if: (i) she is in the highest-bidder set $I$ (where $I\subseteq [n]$) along with bidders $I\setminus\sset{1}$; and (ii) she wins the uniform tie-breaking among bidders $I$. Therefore, her probability of winning is given by (see~\eqref{eq:ex_post_utilities}):
\begin{align}
    &\sum_{\ssets{1}\subseteq I\subseteq [n]}\frac{1}{\card{I}} \prob{\bigland_{i\in I\setminus\ssets{1}}\beta(X_i)=b_j \bigland_{i\in [n]\setminus I}\beta(X_i)<b_j} \notag\\
    &=\sum_{\ssets{1}\subseteq I\subseteq [n]}\frac{1}{\card{I}} \prod_{i\in I\setminus\ssets{1}}\prob{\beta(X_i)=b_j}
    \prod_{i\in [n]\setminus I}\prob{\beta(X_i)<b_j}, &&\text{since $X_i$'s are independent}, \notag\\
     &= \sum_{i=1}^{n} \frac{1}{i} \binom{n-1}{i-1} \prob{\beta(X_1)=b_j}^{i-1} \prob{\beta(X_1)<b_j}^{n-i} , &&\text{since $X_i$'s are identical},\notag\\
     &= \frac{1}{n}\sum_{i=1}^{n} \binom{n}{i} [F(s_{j})-F(s_{j-1})]^{i-1} F(s_{j-1})^{n-i}, \label{eq:utility-deltafunc:helper1}
\end{align}
where for the last equality we used the continuity of the random variable $X_1$, the fact that bid $b_j$ is played exactly in the interval $(s_{j-1},s_j]$, and the combinatorial identity 
\[\frac{1}{i}\binom{n-1}{i-1}
= \frac{1}{i}\frac{(n-1)!}{(n-i)!(i-1)!} 
= \frac{1}{i}\frac{i}{n}\frac{n!}{(n-i)!i!} 
=\frac{1}{n}\binom{n}{i}.\]

Now we consider two cases:
\begin{itemize}
    \item If $F(s_{j})=F(s_{j-1})$, then the probability in~\eqref{eq:utility-deltafunc:helper1} is equal to
    \[\frac{1}{n}\binom{n}{1}\cdot 1\cdot F(s_{j-1})^{n-1} + \frac{1}{n}\sum_{i=2}^{n} \binom{n}{i} \cdot 0\cdot F(s_{j-1})^{n-i}=F(s_{j-1})^{n-1}.\]
    This is indeed equal to the desired expression:
    \[\Delta(s_{j-1},s_j)= \frac{1}{n}\sum_{i=0}^{n-1} F(s_{j-1})^{n-1-i}F(s_j)^i=\frac{1}{n}\sum_{i=0}^{n-1} F(s_{j-1})^{n-1}=F(s_{j-1})^{n-1}.\]
    
    \item If $F(s_{j})\neq F(s_{j-1})$, then we can write the probability in~\eqref{eq:utility-deltafunc:helper1} as:
\begin{align*}
     &\frac{1}{n[F(s_{j})-F(s_{j-1})]} \sum_{i=1}^{n} \binom{n}{i} [F(s_{j})-F(s_{j-1})]^{i} F(s_{j-1})^{n-i} 
    \\  =& \frac{1}{n[F(s_{j})-F(s_{j-1})]} \left[-F(s_{j-1})^n + \sum_{i=0}^{n} \binom{n}{i} \left[F(s_{j})-F(s_{j-1})\right]^{i} F(s_{j-1})^{n-i} \right]
    \\  =& \frac{1}{n[F(s_{j})-F(s_{j-1})]} \left[-F(s_{j-1})^n + \left[\left(F(s_{j})-F(s_{j-1})\right)+F(s_{j-1})\right]^n \right]
    \\  =&\frac{1}{n}\cdot \frac{F(s_j)^n-F(s_{j-1})^n}{F(s_{j})-F(s_{j-1})}\\ 
    =& \frac{1}{n} \sum_{i=0}^{n-1} F(s_{j-1})^{n-1-i}F(s_{j})^i\\
    =& \Delta(s_{j-1},s_j),
\end{align*}
where for the second equality we used the binomial expansion theorem, and for the second to last, the elementary identity $x^n-y^n=(x-y)\sum_{i=0}^{n-1}x^{i}y^{n-1-i}$ for all $x\neq y$ nonnegative reals.
\end{itemize}

\section{Lipschitz Continuity Properties}\label{app:simple-lemma}

In this section we present some standard results about Lipschitz-continuity and apply them to obtain bounds on the Lipschitz constants of some functions that appear in our main arguments. In the following, $\|\cdot\|_\infty$ denotes the standard max-norm operator; that is, for any $\vec{x}=(x_1,x_2,\dots,x_n)\in \R^n$ we have $\|\vec{x}\|\coloneq\max_{j\in[n]} x_j$.

\begin{definition}\label{def:lipschitz}
Let $D \subseteq \mathbb{R}^n$ and $L > 0$. A function $f: D \to \mathbb{R}^m$ is \emph{$L$-Lipschitz (continuous)} if $\|f(x) - f(y)\|_\infty \leq L \|x - y\|_\infty$ for all $x,y \in D$.
\end{definition}

\begin{lemma}\label{lem:lipschitz-composition}
Let $f_1: D_1 \to D_2$ be $L_1$-Lipschitz and $f_2: D_2 \to \mathbb{R}^m$ be $L_2$-Lipschitz. Then $f_2 \circ f_1: D_1 \to \mathbb{R}^m$ is $(L_1 L_2)$-Lipschitz.
\end{lemma}

\begin{proof}
For all $x,y \in D_1$ we have
$\|f_2(f_1(x)) - f_2(f_1(y))\|_\infty \leq L_2 \|f_1(x) - f_1(y)\|_\infty \leq L_2 L_1 \|x - y\|_\infty$.
\end{proof}

\begin{lemma}\label{lem:lipschitz-operations}
Let $f_1: D \to \mathbb{R}$ be $L_1$-Lipschitz and $f_2: D \to \mathbb{R}$ be $L_2$-Lipschitz. Then we have:
\begin{enumerate}
    \item $\alpha \cdot f_1$ is $|\alpha|L_1$-Lipschitz
    \item $f_1 + f_2$ is $(L_1+L_2)$-Lipschitz
    \item $\max\{f_1,f_2\}$ and $\min\{f_1,f_2\}$ are $\max\{L_1,L_2\}$-Lipschitz
    \item $f_1 f_2$ is $M(L_1+L_2)$-Lipschitz, when $\max_{x \in D} |f_1(x)|, \max_{x \in D} |f_2(x)| \leq M$
    \item $1/f_1$ is $M^2L_1$-Lipschitz, when $\max_{x \in D} |1/f_1(x)| \leq M$
\end{enumerate}
\end{lemma}

\begin{proof}
The first and second item are immediate. For the third item, observing that for any $x,y \in D$
$$|\max\{f_1(x),f_2(x)\} - \max\{f_1(y),f_2(y)\}| \leq \max\{|f_1(x)-f_1(y)|, |f_2(x) - f_2(y)|\}$$
and similarly for $\min$, yields the result.
For the fourth item, we have for any $x,y \in D$
\begin{equation*}
\begin{split}
|f_1(x)f_2(x) - f_1(y)f_2(y)| &\leq |f_1(x)f_2(x) - f_1(x)f_2(y)| + |f_1(x)f_2(y) - f_1(y)f_2(y)|\\
&\leq |f_1(x)| |f_2(x) - f_2(y)| + |f_2(y)| |f_1(x) - f_1(y)|\\
&\leq M(L_1+L_2)\|x-y\|_\infty.
\end{split}
\end{equation*}
Finally, for the fifth item, we have for any $x,y \in D$
\begin{equation*}
\begin{split}
|1/f_1(x) - 1/f_1(y)| = \frac{|f_1(y) - f_1(x)|}{|f_1(x)f_1(y)|} \leq M^2L_1 \|x - y\|_\infty. \qedhere
\end{split}
\end{equation*}
\end{proof}

\begin{lemma}\label{lem:lipschitz-deltafunc}
Let $n$ be a positive integer. Then we have
\begin{enumerate}
    \item The function $\phi: [0,1]^2 \to \mathbb{R}, (x,y) \mapsto (1/n) \cdot \sum_{j=0}^{n-1} x^j y^{n-1-j}$ is $n$-Lipschitz.
    \item Let $F$ be $L$-Lipschitz cdf on $[0,1]$. The function $\deltafunc: [0,1]^2 \to \mathbb{R}, (x,y) \mapsto \phi(F(x),F(y))$ is $nL$-Lipschitz.
\end{enumerate}
\end{lemma}

\begin{proof}
The first item is shown using the Lipschitz properties of simple operations from \cref{lem:lipschitz-operations}. The second item then follows using \cref{lem:lipschitz-composition}.
\end{proof}

\begin{lemma}\label{lem:lipschitz-inverse}
Let $f: [a,b] \to \mathbb{R}$ be continuous and $\gamma$-strongly increasing for some $\gamma > 0$ (see \cref{def:strongly-increasing}). Then $f$ is invertible on its image $[c,d] = f([a,b])$ and its inverse $f^{-1}: [c,d] \to [a,b]$ is $(1/\gamma)$-Lipschitz.
\end{lemma}

\begin{proof}
By standard results, it is known that the inverse exists, is continuous, and strictly increasing. Consider any $x, y \in [c,d]$ with $x \geq y$. Since $f^{-1}(x) \geq f^{-1}(y)$ and $f$ is $\gamma$-strongly increasing, we can write
$$f(f^{-1}(x)) \geq f(f^{-1}(y)) + \gamma (f^{-1}(x) - f^{-1}(y))$$
which implies that
$$\frac{1}{\gamma}(x-y) \geq f^{-1}(x) - f^{-1}(y)$$
and thus $f^{-1}$ is $(1/\gamma)$-Lipschitz.
\end{proof}

\begin{lemma}\label{lem:lipschitz-invert-algo}
Let $F$ be an $L$-Lipschitz and $\eps$-strongly increasing cdf on $[0,1]$. Let $\deltafunc$ be as defined in \cref{lem:lipschitz-deltafunc} and let $b_i,b_{i+1} \in [0,1]$ with $b_i < b_{i+1}$.
Define $D := \{(x,y) \in [b_{i+1},1] \times [0,1]| \exists z \in [b_i, x]: (x-b_i) \deltafunc(z,x)=y\}$.
Let $g: D \to [0,1]$ be such that $g(x,y)$ is the unique $z \in [b_i, x]$ satisfying $(x-b_i) \deltafunc(z,x)=y$. Then $g$ is $(L \eps^{-4n} (b_{i+1}-b_i)^{-2n})$-Lipschitz.
\end{lemma}

\begin{proof}
Since $x > b_i$ and $\deltafunc$ is strictly increasing in both arguments, the solution $z$ to the equation $(x-b_i) \deltafunc(z,x)=y$ is indeed unique. Using the definition of $\deltafunc$, we can rewrite the equation as
$$\frac{1}{n} \sum_{j=0}^{n-1} F(z)^j F(x)^{n-1-j} = \frac{y}{x-b_i}.$$
Since $x \geq b_{i+1} > 0$ and $F$ is $\eps$-strongly increasing, we have $F(x) \geq \eps b_{i+1} > 0$. As a result, by dividing by $F(x)^{n-1}$ the equation is equivalent to
$$\frac{1}{n} \sum_{j=0}^{n-1} \left(\frac{F(z)}{F(x)}\right)^j = \frac{y}{(x-b_i)F(x)^{n-1}}$$
which we rewrite as
$$\psi(F(z)/F(x)) = \frac{y}{(x-b_i)F(x)^{n-1}}$$
by defining $\psi: [0,1] \to [0,1], t \to (1/n)\sum_{j=0}^{n-1} t^j$. Finally, we can write
$$g(x,y) = z = F^{-1}\left(F(x) \cdot \psi^{-1}\left(\frac{y}{(x-b_i)F(x)^{n-1}}\right)\right).$$
By \cref{lem:lipschitz-inverse} we have that $F^{-1}$ is $(1/\eps)$-Lipschitz and $\psi^{-1}$ is $n$-Lipschitz. By \cref{lem:lipschitz-operations} the function $x \mapsto (x-b_i)F(x)^{n-1}$ is $nL$-Lipschitz. Note that we have $(x-b_i)F(x)^{n-1} \geq (b_{i+1}-b_i)(\eps b_{i+1})^{n-1} \geq \eps^n (b_{i+1}-b_i)^n$. As a result, the function $(x,y) \mapsto \frac{y}{(x-b_i)F(x)^{n-1}}$ is $(nL \eps^{-2n} (b_{i+1}-b_i)^{-2n})$-Lipschitz. Putting everything together, we obtain that $g$ is $(L \eps^{-4n} (b_{i+1}-b_i)^{-2n})$-Lipschitz.
\end{proof}
    
\end{document}